\documentclass[11pt]{article}
\usepackage{amssymb,latexsym,makeidx,epsfig}
\usepackage{amsthm}
\usepackage{dsfont}
\usepackage{amsmath}
\usepackage{lipsum}
\usepackage{enumerate}
\usepackage{mathrsfs}
\usepackage{xcolor}
\usepackage{tikz}
\usepackage{caption}
\usepackage{setspace}
\usepackage[margin=1in, letterpaper]{geometry}
\usepackage{appendix}
\usepackage{multirow}
\usepackage{subcaption}
\usepackage[colorlinks=true, linkcolor=black, anchorcolor=black, citecolor=black, urlcolor=black, CJKbookmarks=true]{hyperref}
\usepackage{microtype}
\usepackage{colonequals}
\allowdisplaybreaks

\newtheorem{theorem}{Theorem}[section]

\newtheorem{definition}[theorem]{Definition}
\newtheorem{lemma}[theorem]{Lemma}

\newtheorem{proposition}[theorem]{Proposition}
\newtheorem{corollary}[theorem]{Corollary}

\newtheorem{conjecture}[theorem]{Conjecture}
\newtheorem{claim}{Claim}[section]

\numberwithin{figure}{section}
\numberwithin{equation}{section}

\newcommand{\ex}{\operatorname{ex}}
\newcommand{\Rem}{\operatorname{Rem}}
\newcommand{\cut}{\operatorname{cut}}

\begin{document}

\title{NP-hardness of additive approximation for generalized Tur\'{a}n problems}
\author{%
\normalsize
\begin{tabular}{@{}c@{\hspace{0.02\textwidth}}c@{\hspace{0.02\textwidth}}c@{}}
    \begin{minipage}[t]{0.31\textwidth}
        \centering
        Yihan Chen \\
        School of Mathematical Sciences, \\
        University of Science and Technology of China, \\
        Hefei, China \\
        \texttt{cyh2020@mail.ustc.edu.cn}
    \end{minipage}
    &
    \begin{minipage}[t]{0.335\textwidth}
        \centering
        Jialin He \\
        School of Mathematical Sciences, \\
        Key Laboratory of MEA, \\
        Shanghai Key Laboratory of PMMP, \\
        East China Normal University, \\
        Shanghai, China \\
        \texttt{jlhe@math.ecnu.edu.cn}
    \end{minipage}
    &
    \begin{minipage}[t]{0.31\textwidth}
        \centering
        Xihe Li \\
        School of Mathematics and Statistics, \\
        Shaanxi Normal University, \\
        Xi'an, China \\
        \texttt{xiheli@snnu.edu.cn}
    \end{minipage}
    \\[1.5em]
    \multicolumn{3}{@{}c@{}}{%
        \begin{tabular}{@{}c@{\hspace{0.035\textwidth}}c@{}}
            \begin{minipage}[t]{0.48\textwidth}
                \centering
                Jie Ma \\
                School of Mathematical Sciences, \\
                University of Science and Technology of China, \\
                Hefei, China \\
                Yau Mathematical Sciences Center, \\
                Tsinghua University, \\
                Beijing, China \\
                \texttt{jiema@ustc.edu.cn}
            \end{minipage}
            &
            \begin{minipage}[t]{0.48\textwidth}
                \centering
                Tianying Xie \\
                School of Mathematics and Statistics, \\
                Fuzhou University, \\
                Fuzhou, China \\
                \texttt{xiety@fzu.edu.cn}
            \end{minipage}
        \end{tabular}%
    }
\end{tabular}%
}

\date{}
\maketitle

\begin{center}
\begin{minipage}{145mm}

\begin{center}
{\bf Abstract}
\end{center}
{
Given a graph $T$ and a nonempty family $\mathcal{F}$ of graphs, let $\ex(G, T, \mathcal{F})$ denote the maximum number of copies of $T$ in an $\mathcal{F}$-free subgraph of a host graph $G$. Alon and Shikhelman~[{\it Algorithmica}, 2022] conjectured that if $e(T)> 0$ and no member of $\mathcal{F}$ is a subgraph of a blow-up of $T$, then for every $\varepsilon>0$, approximating $\ex(G, T, \mathcal{F})$ within an additive error of $n^{v(T)-\varepsilon}$ is NP-hard for an input graph $G$ on $n$ vertices. 
Previous advances toward this conjecture include the resolution of the canonical case $T=K_2$ by Alon, Shapira, and Sudakov~[{\it Ann. of Math.}, 2009], as well as the case of $\ex(G,K_m,K_k)$ by Alon and Shikhelman.
In this paper, we confirm this conjecture for every bipartite graph $T$ and every clique $T$. 
Our proof reductions combine several techniques, including iterative leaf additions, symmetric rooted extensions, and random transversals. 
Together, they encode the optimal value of the
NP-hard Max $r$-Colorable Subgraph problem into the objectives $\ex(G, T, \mathcal{F})$ for a suitably constructed graph $G$.
\hspace{2em}

}
\end{minipage}
\end{center}

\clearpage

\section{Introduction}
\label{sec:introduction}

Graph modification problems ask for the minimum number of modifications needed to transform an input graph $G$ into a graph satisfying some predefined property $\mathcal{P}$.
Common modifications include adding or deleting edges or vertices.
The systematic study of such problems was initiated by Yannakakis~\cite{Yan978,Yan981}, who proved that the edge-deletion problem is NP-hard for several natural graph properties, such as outerplanarity, transitive orientability, and certain forbidden-subgraph properties.
Apart from classical graph-theoretic aspects, these problems also arise in numerous practical applications, including machine learning~\cite{BaBC,FrKN,Vel}, molecular biology~\cite{GoKS}, epidemic control~\cite{GaMa,YSPJ}, numerical algebra~\cite{Rose}, communication network design~\cite{XuSa}, and circuit design~\cite{ElCo}.

In~\cite{Yan981}, Yannakakis raised the question of proving NP-hardness for general families of properties.
For a family $\mathcal{F}$ of graphs,\footnote{Throughout the paper, unless stated otherwise, a family of graphs may be finite or infinite; all graph families are nonempty, and every graph therein has at least one edge.} we say that a graph is {\it $\mathcal{F}$-free} if it contains no member of $\mathcal{F}$ as a (not necessarily induced) subgraph.
For a graph $G=(V,E)$, let $v(G)\colonequals |V(G)|$ and $e(G)\colonequals |E(G)|$ denote its numbers of vertices and edges, respectively. Let $\chi(G)$ denote the chromatic number of $G$, and let $\omega(G)$ denote its clique number, that is, the order of a maximum complete subgraph in $G$. Let $\Rem_{\mathcal{F}}(G)$ be the minimum number of edges whose deletion turns $G$ into an $\mathcal{F}$-free graph.
A major breakthrough of Alon, Shapira and Sudakov~\cite{AlSS} addresses Yannakakis's question for the property of $\mathcal{F}$-freeness.
The authors showed that for every nonempty family $\mathcal{F}$ of graphs, if $\mathcal{F}$ contains no bipartite graphs, then it is NP-hard to compute $\Rem_{\mathcal{F}}(G)$, and even to approximate it within an additive error of $v(G)^{2-\varepsilon}$ for any fixed $\varepsilon>0$.
For more related results on this problem, we refer the interested reader to \cite{AlSt,ArSS,BCKP,GiHo,GiLS24,GiLS25}.
There have also been recent developments in computational complexity concerning graph modification~\cite{ADIT}, subgraph counting~\cite{BoBB,DoMW,LiPY}, and forbidden subgraphs~\cite{ChLN,JMOPPSL,KoLo,LMPS}.

Note that if $G=K_n$ is the complete graph on $n$ vertices, then computing $\Rem_{\mathcal{F}}(K_n)$ is equivalent to determining the Tur\'{a}n number $\ex(n, \mathcal{F})$, where $\ex(n, \mathcal{F})$ denotes the maximum number of edges of an $n$-vertex $\mathcal{F}$-free graph.
Indeed, we have $\ex(n, \mathcal{F})={n\choose 2}-\Rem_{\mathcal{F}}(K_n).$
In~\cite{AlSh16}, Alon and Shikhelman initiated the systematic study of generalized Turán numbers, and in~\cite{AlSh22} they studied the corresponding optimization problem for the objective functions defined below.

\begin{definition}\label{def:ex}
{\rm
Let $G$ and $T$ be graphs and $\mathcal{F}$ be a family of graphs.
\begin{itemize}
\item[(i)] Define $\ex(G, T, \mathcal{F})$ as the maximum number of copies of $T$ in an $\mathcal{F}$-free subgraph of $G$.
\item[(ii)] Define $\ex'(G, T, \mathcal{F})$ as the maximum number of copies of $T$ in an $\mathcal{F}$-free spanning subgraph of $G$.
\end{itemize}
}
\end{definition}

\noindent Throughout the complexity statements in this paper, the graph $T$, the family $\mathcal{F}$, and the approximation parameter $\varepsilon$ are fixed; only the host graph $G$ is part of the input.
When $\mathcal{F}$ consists of a single graph $F$, we write $\ex(G, T, F)$ or $\ex'(G, T, F)$ for short.
Before proceeding, we point out some simple yet useful facts.  
First, $\ex'(G, T, \mathcal{F})\leq \ex(G, T, \mathcal{F})$ always holds. 
For every finite family $\mathcal{F}$, we have $\ex(G, T, \mathcal{F})-\ex'(G, T, \mathcal{F})$ $=O_{T,\mathcal{F}}(1)$ (see Proposition~\ref{prop:ex}~(i)).
Moreover, if no member of $\mathcal{F}$ has isolated vertices, then $\ex(G, T, \mathcal{F})=\ex'(G, T, \mathcal{F})$ for every graph $G$, even when $\mathcal{F}$ is infinite  (see Proposition~\ref{prop:ex}~(ii)).

Using this terminology, the result of Alon, Shapira and Sudakov~\cite{AlSS} can be restated as follows: for every nonempty family $\mathcal{F}$ of graphs containing no bipartite graphs, it is NP-hard to approximate $\ex'(G, K_2, \mathcal{F})$ within an additive error of $v(G)^{2-\varepsilon}$ for any fixed $\varepsilon>0$.
Their reduction also establishes hardness for $\ex(G, K_2, \mathcal{F})$ (see \cite[Theorem~6.1 and Section~8]{AlSS}).
Following the spirit of this result, Alon and Shikhelman~\cite{AlSh22} raised the following natural questions.
Given a graph $T$, a family $\mathcal{F}$ of graphs, and a constant $\varepsilon>0$, is there a polynomial-time algorithm that determines $\ex(G, T, \mathcal{F})$ within an additive error of $\varepsilon\cdot v(G)^{v(T)}$ for any input graph $G$?
If so, is the problem of finding a significantly better approximation NP-hard?
For the first question, they gave an affirmative answer for every finite family $\mathcal{F}$.
For the question of better approximation, they obtained the following results.
For a fixed graph $T$, a {\it blow-up} of $T$ is the graph obtained by replacing each vertex of $T$ by an independent set and each edge by a complete bipartite graph between the corresponding independent sets.

\begin{theorem}[\cite{AlSh22}]\label{thm:AlSh22-1}
Let $T$ be a fixed graph and let $\mathcal{F}$ be a family of graphs such that there is a graph $F\in \mathcal{F}$ which is a subgraph of a blow-up of $T$.
Then there is a fixed $\varepsilon=\varepsilon(T, \mathcal{F})>0$ such that $\ex(G, T, \mathcal{F})$ can be approximated in polynomial time within an additive error of $n^{v(T)-\varepsilon}$ for any input graph $G$ on $n$ vertices.
\end{theorem}

\begin{theorem}[\cite{AlSh22}]\label{thm:AlSh22-2}
Let $k>m\geq 2$ be integers.
Then for every $\varepsilon>0$, approximating $\ex(G, K_m, K_k)$ within an additive error of $n^{m-\varepsilon}$ is NP-hard for an input graph $G$ on $n$ vertices.
\end{theorem}

Since $K_k$ has no isolated vertices, Proposition~\ref{prop:ex}~(ii) implies that $\ex(G, K_m, K_k)=\ex'(G, K_m, K_k)$.
In light of these results, they proposed the following conjecture.

\begin{conjecture}[\cite{AlSh22}]\label{conj:AlSh22}
Let $T$ be a graph with $e(T)> 0$, and $\mathcal{F}$ be a nonempty family such that no $F\in \mathcal{F}$ is a subgraph of a blow-up of $T$.
Then for every $\varepsilon>0$, approximating $\ex(G, T, \mathcal{F})$ within an additive error of $n^{v(T)-\varepsilon}$ is NP-hard for an input graph $G$ on $n$ vertices.
\end{conjecture}

Before the present work, Conjecture~\ref{conj:AlSh22} had been established for the canonical case $\ex(G, K_2, \mathcal{F})$ (the aforementioned result of Alon, Shapira and Sudakov~\cite{AlSS}) and for $\ex(G, K_m, K_k)$ (Theorem~\ref{thm:AlSh22-2}).
Alon and Shikhelman also proved hardness at additive error $n^{v(T)-2-\varepsilon}$ when $T$ and all graphs in $\mathcal{F}$ are 3-connected (see \cite[Proposition 6.1]{AlSh22}).

\medskip

In this paper, we make progress on Conjecture~\ref{conj:AlSh22}, substantially extending the results of \cite{AlSS} and \cite{AlSh22} mentioned above.
Our first main result confirms this conjecture for every bipartite graph $T$.
While maximization over arbitrary and spanning subgraphs may differ when forbidden graphs have isolated vertices, our reductions apply to both conventions.

\begin{theorem}\label{thm:bipar}
Let $T$ be a bipartite graph with $e(T)> 0$, $\mathcal{F}$ be a nonempty family such that no $F\in \mathcal{F}$ is a subgraph of a blow-up of $T$, and $\varepsilon>0$.
Then for each of the two parameters $\ex(G, T, \mathcal{F})$ and $\ex'(G, T, \mathcal{F})$, approximation within an additive error of $n^{v(T)-\varepsilon}$ is NP-hard for an input graph $G$ on $n$ vertices.
\end{theorem}

We also confirm Conjecture~\ref{conj:AlSh22} for several additional classes of pairs $(T, \mathcal{F})$.
For positive integers $s_1, \ldots, s_r$, let $K_{s_1, \ldots, s_r}$ be the complete $r$-partite graph with partite sets of sizes $s_1, \ldots, s_r$.

\begin{theorem}\label{thm:multipar}
Let $T=K_{s_1, \ldots, s_r}$ with $r\geq 3$ and $s_1, \ldots, s_r\geq 1$, $\mathcal{F}$ be a nonempty family such that no $F\in \mathcal{F}$ is a subgraph of a blow-up of $T$, and $\varepsilon>0$.
Then for each of the two parameters $\ex(G, T, \mathcal{F})$ and $\ex'(G, T, \mathcal{F})$, approximation within an additive error of $n^{v(T)-\varepsilon}$ is NP-hard for an input graph $G$ on $n$ vertices.
\end{theorem}

Theorems~\ref{thm:bipar} and \ref{thm:multipar} together confirm Conjecture~\ref{conj:AlSh22} for every complete multipartite graph $T$, and in particular for every clique $T$.
We also confirm the conjecture whenever $\mathcal{F}=\{K_k\}$ consists of a single clique (see Theorem~\ref{thm:clique}).
These results extend Theorem~\ref{thm:AlSh22-2} on $\ex(G, K_m, K_k)$ in two directions: either relaxing $\mathcal{F}=\{K_k\}$ to a general forbidden family, or relaxing $K_m$ to a general counted graph.

Finally, we establish Conjecture~\ref{conj:AlSh22} whenever every graph in $\mathcal{F}$ has chromatic number at least $\chi(T)+2$. 
It follows from a more general form (Theorem~\ref{thm:chromatic}).

\begin{theorem}\label{thm:chromatic+2}
Let $T$ be a graph with $e(T)>0$, $\mathcal{F}$ be a nonempty family of graphs with $\min_{F\in \mathcal{F}}\chi(F)\geq \chi(T)+2$, and $\varepsilon>0$.
Then for each of the two parameters $\ex(G, T, \mathcal{F})$ and $\ex'(G, T, \mathcal{F})$, approximation within an additive error of $n^{v(T)-\varepsilon}$ is NP-hard for an input graph $G$ on $n$ vertices.
\end{theorem}

The remainder of this paper is organized as follows.
In the next section, we introduce some additional terminology and provide a proof sketch of our results.
In Section~\ref{sec:pre}, we state and prove several results that will be used in our proofs.
In Section~\ref{sec:pf_bipar}, we present our proof of Theorem~\ref{thm:bipar}.
In Section~\ref{sec:pf_nonbipar}, we prove Theorems~\ref{thm:multipar} and \ref{thm:chromatic+2}.
Finally, we conclude this paper with some remarks and further problems in Section~\ref{sec:conclu}.
\vspace{0.2cm}

\noindent{\bf Note added.} After the first version of this paper was posted on arXiv, we learned that Gishboliner, Levanzov and Shapira~\cite{GiLS26} had independently verified Conjecture~\ref{conj:AlSh22} when $T$ is a clique and $\mathcal{F}$ is a finite family of graphs, using a different method.

\section{Terminology and proof sketch}
\label{sec:sketch}

In this section, we first introduce some necessary terminology.
Then we provide a proof sketch of our results in Section~\ref{subsec:sketch}.

In this paper, all graphs are finite and simple.
Given two graphs $G$ and $H$, we write $G\cup H$ for the disjoint union of $G$ and $H$, and $mG$ for the disjoint union of $m$ copies of $G$.
For any vertex $v\in V(G)$, let $N_G(v)$ be the neighborhood of $v$ in $G$, and $d_G(v)\colonequals|N_G(v)|$ be the degree of $v$.
Given a subset $U\subseteq V(G)$, let $G[U]$ be the subgraph of $G$ induced by $U$, and let $G-U$ be shorthand for $G[V(G)\setminus U]$.
We use $\delta(G)$ to denote the minimum degree of $G$.
For any graph $G'$ with $V(G')\subseteq V(G)$, let $G+G'$ be the graph with vertex set $V(G)$ and edge set $E(G)\cup E(G')$.
In the special case when $G'$ consists of a single edge $e$, we will simply write it as $G+e$.

For two graphs $H$ and $T$, we use $T\subseteq H$ to denote that $T$ is a subgraph of $H$.
A {\it copy} of $T$ in $H$ is a not necessarily induced subgraph of $H$ isomorphic to $T$.
Copies with different vertex sets or different edge sets are counted separately.
We write $\mathcal{N}(H,T)$ for their number.
Then
\begin{align*}
  &~\ex(G, T, \mathcal{F})=\max\{\mathcal{N}(H,T)\colon\, H\subseteq G, \mbox{$H$ is $\mathcal{F}$-free}\}, \\
  &~\ex'(G, T, \mathcal{F})=\max\{\mathcal{N}(H,T)\colon\, H\subseteq G, V(H)=V(G), \mbox{$H$ is $\mathcal{F}$-free}\}.
\end{align*}
We call a (not necessarily injective) map $\varphi: V(T)\to V(H)$ a {\it homomorphism} if it preserves all adjacencies, i.e., $\varphi(u)\varphi(v)\in E(H)$ for every $uv\in E(T)$.
We will also use the shorthand $\varphi: T\to H$ to denote a homomorphism from $T$ to $H$.
Denote by $\hom(T, H)$ the number of homomorphisms from $T$ to $H$.
Note that $F$ is a subgraph of a blow-up of $T$ if and only if there exists a homomorphism from $F$ to $T$.

Let $G$ be a graph and let $U\subseteq V(G)$ be a nonempty set.
By {\it identifying the vertices of $U$}, we mean replacing them by a single new vertex $u$, where a vertex $v\in V(G)\setminus U$ is adjacent to $u$ if and only if $N_G(v)\cap U\neq \varnothing$, and edges with both endpoints in $V(G)\setminus U$ are preserved.
Edges with both endpoints in $U$ are discarded, and parallel edges are replaced by a single edge, so that the resulting graph is a simple graph.

For a fixed graph $G$ and a positive integer $m$, the {\it $m$-blow-up} of $G$, denoted by $G(m)$, is the graph obtained by replacing each vertex of $G$ by an independent set of size $m$ and each edge by a complete bipartite graph between the corresponding independent sets.
We call each set of vertices obtained by blowing up a vertex of $G$ a {\it cluster}.
For each $v\in V(G)$, let $C_v$ denote the cluster corresponding to $v$.
Unless a common cluster size is specified, a blow-up may have unequal positive cluster sizes.

Throughout this paper, we may assume that $0<\varepsilon\leq 1$ when proving the hardness results.
Indeed, if $\varepsilon>1$, then an algorithm approximating $\ex(G,T,\mathcal{F})$ within an additive error of $v(G)^{v(T)-\varepsilon}$ also approximates it within an additive error of $v(G)^{v(T)-1}$.
The same observation applies to $\ex'(G,T,\mathcal{F})$.
We use the standard asymptotic notation $O$, $o$, and $\Theta$.
Subscripts specify the permitted dependence of the implicit constants.
For example, for a nonnegative function $f$, the assertion $h(n)=O_T(f(n))$ means that there are constants $C_T>0$ and $n_T$, depending only on $T$, such that $|h(n)|\leq C_Tf(n)$ for all $n\geq n_T$.
For an event $A$, let $\Pr[A]$ be the probability of $A$.
For a random variable $X$, we use $\mathbb{E}[X]$ to denote the expectation of $X$.

\subsection{Proof sketch}
\label{subsec:sketch}

In our proof, we reduce the following \textsc{Max $r$-Colorable Subgraph} problem to the problem of approximating $\ex(G, T, \mathcal{F})$ or $\ex'(G, T, \mathcal{F})$ within an additive error of $n^{v(T)-\varepsilon}$.
For a graph $G$ and an integer $r\geq 2$, let $$\cut_r(G)\colonequals \max\{e(H)\colon\, H\subseteq G, V(H)=V(G), \chi(H)\leq r\}.$$ The problem of determining $\cut_r(G)$ is NP-hard for every fixed $r\geq 2$ (see~\cite{Vit}).

Let $k\colonequals \min_{F\in \mathcal{F}}\chi(F)$.
We restrict the source instances $G$ to a dense class $\mathcal{G}_k$ such that on this class, computing $z\colonequals \cut_{k-1}(G)$ remains NP-hard and $z=\ex(G,K_2,K_k)$.
The class $\mathcal{G}_k$ is guaranteed by Lemma~\ref{le:AlSS-NPhard}.
Our goal is to construct a graph $Z$ from $G$ in polynomial time so that
\begin{equation}\label{eq:bounds-sketch}
L+Sz-a \leq \ex'(Z, T, \mathcal{F})\leq \ex(Z, T, \mathcal{F})\leq L+Sz+b,
\end{equation}
where $L=L(G)\geq 0$ and $S=S(G)>0$ are polynomial-time computable rational numbers, and $a,b$ are nonnegative errors.
For each fixed $\varepsilon>0$, we choose the construction parameters so that $a+b+v(Z)^{v(T)-\varepsilon}<S/2$ for every sufficiently large graph $G$.
Since $z=\cut_{k-1}(G)$ is an integer, an approximation to either $\ex(Z,T,\mathcal{F})$ or $\ex'(Z,T,\mathcal{F})$ with additive error at most $v(Z)^{v(T)-\varepsilon}$ determines $z$ exactly by subtracting $L$, dividing by $S$, and rounding to the nearest integer (see Corollary~\ref{cor:transfer}).

The main difficulty lies in constructing $Z$ and establishing the upper and lower bounds in~\eqref{eq:bounds-sketch}.
We use different constructions according to the structural properties of $T$ and $\mathcal{F}$.
For bipartite graphs $T$, we use \emph{iterative leaf additions} (see Figure~\ref{fig:leaf}) followed by a blow-up.
For Theorem~\ref{thm:chromatic+2}, we use \emph{symmetric rooted extensions} (see Figure~\ref{fig:root}).
For Theorem~\ref{thm:multipar}, we use a related construction adapted to complete multipartite graphs.
A common ingredient in the upper bounds is a \emph{random-transversal argument}.
We first sketch the proof for bipartite graphs and then discuss the remaining cases.

\begin{figure}[htbp]
\centering
\begin{tikzpicture}[
  core/.style={circle, draw=black, fill=black!80, inner sep=0pt, minimum size=5pt},
  firstleaf/.style={circle, draw=blue!65!black, fill=blue!15, inner sep=0pt, minimum size=4.5pt},
  secondleaf/.style={circle, draw=red!70!black, fill=red!15, inner sep=0pt, minimum size=4.5pt},
  edge/.style={draw=black!60, line width=0.5pt},
  every node/.style={font=\small}
]

\begin{scope}[yshift=4cm]
  \node[font=\small\bfseries] at (0,-0.5) {$L_4(P_3)$};

  \draw[edge] (-5,0)--(0,0)--(5,0);
  \foreach \x in {-5,0,5} {
    \begin{scope}[xshift=\x cm]
      \foreach \d in {-1.5,-0.5,0.5,1.5} {
        \draw[edge] (0,0)--(\d,1.15);
      }
    \end{scope}
  }

  \foreach \x in {-5,0,5} {
    \begin{scope}[xshift=\x cm]
      \node[core] at (0,0) {};
      \foreach \d in {-1.5,-0.5,0.5,1.5} {
        \node[firstleaf] at (\d,1.15) {};
      }
    \end{scope}
  }
\end{scope}

\begin{scope}
  \node[font=\small\bfseries] at (0,-1.45)
    {$L_2(L_4(P_3))$};

  \draw[edge] (-5,0)--(0,0)--(5,0);
  \foreach \x in {-5,0,5} {
    \begin{scope}[xshift=\x cm]
      \foreach \d in {-1.5,-0.5,0.5,1.5} {
        \draw[edge] (0,0)--(\d,1.15);
      }

      \foreach \d/\ell/\r in {
        -1.5/-1.65/-1.35,
        -0.5/-0.65/-0.35,
         0.5/ 0.35/ 0.65,
         1.5/ 1.35/ 1.65
      } {
        \draw[edge] (\d,1.15)--(\ell,2.55);
        \draw[edge] (\d,1.15)--(\r,2.55);
      }
    \end{scope}
  }

  \draw[edge] (0,0)--(-0.15,-0.95);
  \draw[edge] (0,0)--( 0.15,-0.95);

  \foreach \x in {-5,5} {
    \begin{scope}[xshift=\x cm]
      \draw[edge] (0,0)--(-0.15,-0.95);
      \draw[edge] (0,0)--( 0.15,-0.95);
    \end{scope}
  }

  \foreach \x in {-5,0,5} {
    \begin{scope}[xshift=\x cm]
      \node[core] at (0,0) {};
      \foreach \d in {-1.5,-0.5,0.5,1.5} {
        \node[firstleaf] at (\d,1.15) {};
      }

      \foreach \ell/\r in {
        -1.65/-1.35,
        -0.65/-0.35,
         0.35/ 0.65,
         1.35/ 1.65
      } {
        \node[secondleaf] at (\ell,2.55) {};
        \node[secondleaf] at (\r,2.55) {};
      }
    \end{scope}
  }

  \node[secondleaf] at (-0.15,-0.95) {};
  \node[secondleaf] at ( 0.15,-0.95) {};

  \foreach \x in {-5,5} {
    \begin{scope}[xshift=\x cm]
      \node[secondleaf] at (-0.15,-0.95) {};
      \node[secondleaf] at ( 0.15,-0.95) {};
    \end{scope}
  }
\end{scope}

\end{tikzpicture}
\caption{An illustration of iterative leaf additions with $G=P_3$, where $P_3$ denotes the path on three vertices. We attach four leaves to each vertex of $P_3$ to obtain $L_4(P_3)$, and then attach two leaves to each vertex of $L_4(P_3)$ to obtain $L_2(L_4(P_3))$.}
\label{fig:leaf}
\end{figure}
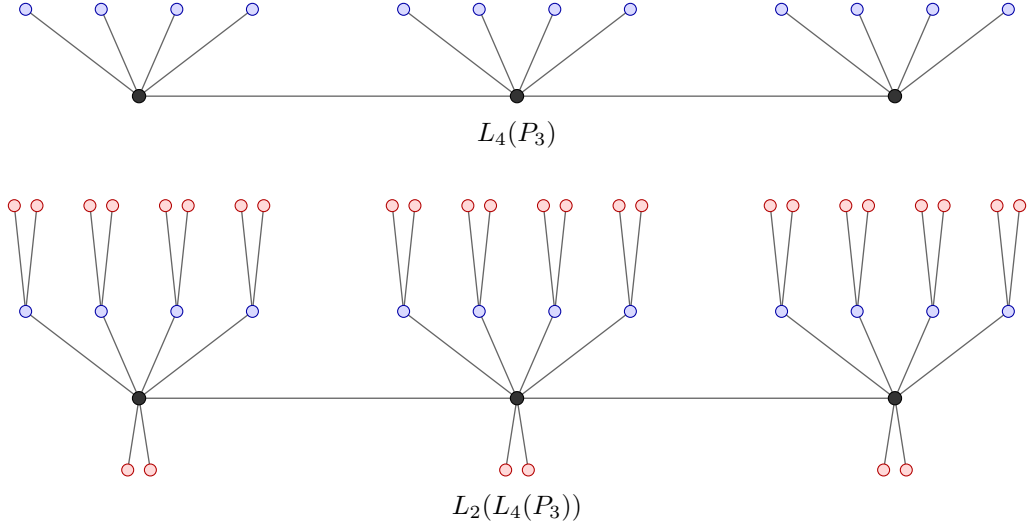

\medskip
\noindent\textbf{Bipartite graphs $T$ (Theorem~\ref{thm:bipar}).}
We complete the proof in the following three steps.

\medskip
\noindent\emph{Step 1: Iterative leaf additions and homomorphism counts.}
Given an $n$-vertex graph $G$ and a positive integer $x$, let $L_x(G)$ be the graph obtained from $G$ by attaching $x$ leaves to every vertex of $G$, where all added vertices are distinct and lie outside $V(G)$.
Let $s\colonequals v(T)-2$ and $\mathbf{x}=(x_1, \ldots, x_s)$ be a vector of $s$ positive integers.
Define $$L_{\mathbf{x}}(G)\colonequals L_{x_1}\left(L_{x_2}\left(\cdots L_{x_s}\left(G\right) \cdots\right)\right).$$
Repeated application of the following quotient operation expresses $\hom(T,L_{\mathbf{x}}(G))$ as a linear combination of the numbers $\hom(Q,G)$, with polynomial coefficients in $x_1,\ldots,x_s$; see~\eqref{eq:hom-sketch} below.
For a bipartite graph $Q$ and an independent set $I\subseteq V(Q)$, the quotient graph $Q_I$ is constructed from $Q$ as follows:
(1) Delete the vertices of $I$;
(2) For each nonisolated vertex $v\in I$, identify all vertices of $N_Q(v)$ to form a single vertex $v'$, taking the transitive closure of these identifications;
(3) For each isolated vertex $v\in I$, add one new isolated vertex $v'$.
Then let $(I_1,\ldots,I_s)$ be a sequence of vertex sets and $(T_0, T_1, \ldots, T_s)$ be a sequence of graphs obtained as follows:
\begin{equation}\label{eq:Ti-sketch}
T_0\colonequals T, \qquad
I_i\text{ is an independent set of }T_{i-1},\qquad
T_i=(T_{i-1})_{I_i}
\end{equation}
for every $i=1, \ldots, s$.
Using the quotient operation, we derive that
\begin{equation}\label{eq:hom-sketch}
\hom(T, L_{\mathbf{x}}(G))=\sum\limits_{Q\in \mathcal{Q}_T}f_Q(\mathbf{x})\hom(Q, G),
\end{equation}
where $\mathcal{Q}_T$ is the set of bipartite graphs on at most $v(T)$ vertices, and
$$f_Q(\mathbf{x})= \sum_{\substack{\text{\rm $(I_1,\ldots,I_s)$ satisfying~\eqref{eq:Ti-sketch},}\\ T_s\cong Q}}\left(\prod_{i=1}^s x_i^{|I_i|}\right).$$

For a graph $Q$, let $p(Q)$ be the number of nonisolated vertices of $Q$.
Note that we have either $p(Q)=0$ or $p(Q)\geq 2$.
Then \eqref{eq:hom-sketch} becomes
\begin{equation}\label{eq:hom-3part-sketch}
\hom(T, L_{\mathbf{x}}(G))=A_0(\mathbf{x},G)+A_2(\mathbf{x},G)+A_{\geq 3}(\mathbf{x},G),
\end{equation}
where $A_0(\mathbf{x},G)$, $A_2(\mathbf{x},G)$ and $A_{\geq 3}(\mathbf{x},G)$ denote the sums of the terms in \eqref{eq:hom-sketch} with $p(Q)=0$, $p(Q)=2$ and $p(Q)\geq 3$, respectively.
If $p(Q)=2$, then $Q$ consists of a single edge together with isolated vertices, so $A_2(\mathbf{x},G)$ is proportional to $e(G)$.
The term $A_0(\mathbf{x},G)$ depends on $G$ only through $n$ and contributes to the computable number $L$.
We choose $x_1,\ldots,x_s$ carefully so that $A_{\geq 3}(\mathbf{x},G)$ is small relative to the coefficient of $e(G)$ in $A_2(\mathbf{x},G)$.
More precisely, we have
\begin{equation}\label{eq:A-sketch}
A_0(\mathbf{x},G)=\sum_{\substack{Q\in \mathcal{Q}_T,\\ p(Q)=0}}f_Q(\mathbf{x})n^{v(Q)}, \quad A_2(\mathbf{x},G)=S(\mathbf{x},n)\cdot e(G), \quad A_{\geq 3}(\mathbf{x},G)=S(\mathbf{x},n)\cdot O_T\left(n^{-3}\right),
\end{equation}
where $S(\mathbf{x},n) = \sum\nolimits_{Q\in \mathcal{Q}_T,~p(Q)=2} 2\cdot f_Q(\mathbf{x})n^{v(Q)-2}$ is an integer greater than $1$.
The key point is that sequences with lexicographically maximal $(|I_1|,\ldots,|I_s|)$ among those with $e(T_s)>0$ satisfy $p(T_s)=2$.
Then the choice in~\eqref{eq:leaves} makes earlier exponents dominate later ones, which proves~\eqref{eq:A-sketch} (see Lemmas~\ref{le:K2} and \ref{le:leaf-hom-3}).

\medskip
\noindent\emph{Step 2: Blow-up and the random-transversal argument.}
Let $Z(G)= L_{\mathbf{x}}(G)(B)$ be the $B$-blow-up of $L_{\mathbf{x}}(G)$.
Here $B$ is carefully chosen so that we can control the errors in the following three places:
\begin{itemize}
\item[(1)] the random-transversal error in $\mathbb{E}[e(R_J)]\leq z+O(n^{k-\eta}B^{-\eta})$ (see~\eqref{eq:eRJ-sketch}),
\item[(2)] the error for noninjective homomorphisms when we pass from homomorphism counts to copies of $T$ in Step~3, and
\item[(3)] the error $v(Z)^{v(T)-\varepsilon}$ when we finally deduce $a+b+v(Z)^{v(T)-\varepsilon}<S/2$.
\end{itemize}

We first derive an upper bound for homomorphism counts.
For each $v\in V(L_{\mathbf{x}}(G))$, let $C_v$ be the cluster corresponding to $v$ in $Z$.
For any $\mathcal{F}$-free subgraph $J\subseteq Z$, define two graphs $J^-$ and $J^{+}$ as follows.
Let $J^-$ be the subgraph of $J$ induced by the vertices lying in clusters corresponding to $V(G)$.
Let $J^+$ be the spanning subgraph of $Z$ obtained from $J$ by adding every missing vertex of $Z$ and every edge of $Z$ that does not belong to the source subgraph $G(B)$.
Then $J^-\subseteq J\subseteq J^{+}\subseteq Z$, and in particular, $\hom(T, J)\leq \hom(T, J^{+})\leq \hom(T, Z)$.
In order to derive an upper bound on $\hom(T, J)$, we define a random spanning subgraph $R_J$ of $G$:
For each $v\in V(G)$, choose one vertex $v'\in C_v$ independently and uniformly.
The random graph $R_J$ is defined by
$$V(R_J)\colonequals V(G), \qquad E(R_J)\colonequals \{uv\in E(G)\colon\, u',v'\in V(J),\ u'v'\in E(J)\}.$$
Thus $R_J$ is obtained by retaining the edges of $J$ between the selected vertices and relabelling each selected vertex $v'$ as $v$.
We now use properties of $R_J$ to show that
\begin{equation}\label{eq:hom-upper-sketch}
\hom(T, J)\leq B^{v(T)}\left(A_0(\mathbf{x},G) +S(\mathbf{x},n)\cdot z +S(\mathbf{x},n)\cdot O\left(n^{k-\eta}B^{-\eta} +n^{-3}\right)\right)
\end{equation}
for some constant $\eta>0$.
The coefficient $B^{v(T)}$ follows from the fact that every homomorphism from $T$ to $Z$ projects to a unique homomorphism from $T$ to $L_{\mathbf{x}}(G)$, and each homomorphism into $L_{\mathbf{x}}(G)$ has exactly $B^{v(T)}$ lifts into $Z$.
The terms $A_0(\mathbf{x},G)$ and $S(\mathbf{x},n)\cdot O\left(n^{-3}\right)$ correspond to the terms $A_0(\mathbf{x},G)$ and $A_{\geq 3}(\mathbf{x},G)$ in~\eqref{eq:hom-3part-sketch}.
We now derive the terms $S(\mathbf{x},n)\cdot z+S(\mathbf{x},n)\cdot O\left(n^{k-\eta}B^{-\eta}\right)$ from the $A_2(\mathbf{x},G)$-part.
For this part, we in fact count homomorphisms from $T$ to $J^{+}$ since $\hom(T, J)\leq \hom(T, J^{+})$ and we only need an upper bound.
Recall that the $A_2(\mathbf{x},G)$-part corresponds to graphs consisting of a single edge together with isolated vertices.
Combining this with properties of $R_J$, we first deduce an upper bound of the form $S(\mathbf{x},n)\cdot \mathbb{E}[e(R_J)]$ (see Lemma~\ref{le:leaf-hom-4}).
Now it suffices to prove that
\begin{equation}\label{eq:eRJ-sketch}
\mathbb{E}[e(R_J)]\leq z+O(n^{k-\eta}B^{-\eta}).
\end{equation}
Since deleting one edge from each copy of $K_k$ in $R_J$ yields a $K_k$-free subgraph of $G$, we have $e(R_J)\leq \ex(G, K_2, K_k)+\mathcal{N}(R_J, K_k)=z+\mathcal{N}(R_J, K_k).$
Then it remains to derive an upper bound on $\mathbb{E}[\mathcal{N}(R_J, K_k)]$.
Fix an $F_0\in\mathcal{F}$ with $\chi(F_0)=k$.
Then $F_0$ is a subgraph of a blow-up of $K_k$.
Moreover, since $J$ is $\mathcal{F}$-free and $J^{-}\subseteq J$, we know that $J^{-}$ is $F_0$-free.
By the counting estimate of Alon and Shikhelman~\cite{AlSh16} (see Lemma~\ref{le:AlSh16-degen}), we have $\mathcal{N}(J^{-},K_k)\leq cv(J^{-})^{k-\eta}\leq c(nB)^{k-\eta}$ for some constants $c>0$ and $0<\eta<k$.
Since for each copy of $K_k$ in $J^{-}$, the probability that all its vertices are selected is $B^{-k}$, we have $\mathbb{E}[\mathcal{N}(R_J, K_k)]=B^{-k}\mathcal{N}(J^{-},K_k)\leq cn^{k-\eta}B^{-\eta}.$
This completes the derivation of~\eqref{eq:hom-upper-sketch}.

For every spanning subgraph $H\subseteq G$, we also have the following lower bound estimate for homomorphism counts:
\begin{equation}\label{eq:hom-lower-sketch}
\hom(T, L_{\mathbf{x}}(H)(B))\geq B^{v(T)}\Big(A_0(\mathbf{x}, H)+S(\mathbf{x},n)\cdot e(H)\Big).
\end{equation}
Indeed, the coefficient $B^{v(T)}$ follows from the fact that each homomorphism has exactly $B^{v(T)}$ lifts,
and the terms $A_0(\mathbf{x}, H)$ and $S(\mathbf{x},n)\cdot e(H)$ follow from~\eqref{eq:hom-3part-sketch}, \eqref{eq:A-sketch} and $A_{\geq 3}(\mathbf{x},H)\geq 0$.

\medskip
\noindent\emph{Step 3: From homomorphisms to copies and completing the proof.}
We use the following estimate to pass from homomorphism counts to copies of $T$:
\begin{equation}\label{eq:aut-sketch}
0\leq \hom(T, A)-\gamma\mathcal{N}(A, T)\leq v(T)^{2}N^{v(T)-1}
\end{equation}
for any $N$-vertex graph $A$, where $\gamma$ is the number of automorphisms of $T$.
For the lower bound in~\eqref{eq:bounds-sketch}, we choose a spanning subgraph $H$ of $G$ with $\chi(H)\leq k-1$ and $e(H)=z$.
Then $L_{\mathbf{x}}(H)(B)$ is a spanning subgraph of $Z=L_{\mathbf{x}}(G)(B)$.
Moreover, we have $\chi(L_{\mathbf{x}}(H)(B))\leq k-1$ since a proper $(k-1)$-coloring of $H$ extends to $L_{\mathbf{x}}(H)(B)$ through each leaf-addition step and then to the blow-up.
Thus $L_{\mathbf{x}}(H)(B)$ is $\mathcal{F}$-free by $\min_{F\in \mathcal{F}}\chi(F)=k>\chi(L_{\mathbf{x}}(H)(B))$.
Furthermore, since $H$ and $G$ have the same vertex set, we have $A_0(\mathbf{x},H)=A_0(\mathbf{x},G)$ by~\eqref{eq:A-sketch}.
Combining \eqref{eq:hom-3part-sketch}, \eqref{eq:A-sketch}, \eqref{eq:hom-lower-sketch} and \eqref{eq:aut-sketch}, we have
\begin{equation}\label{eq:lower-sketch}
\ex'(Z, T, \mathcal{F}) \geq \mathcal{N}(L_{\mathbf{x}}(H)(B), T) \geq \gamma^{-1}B^{v(T)}\big(A_0(\mathbf{x},G)+S(\mathbf{x},n)\cdot z\big)-\gamma^{-1}v(T)^2v(Z)^{v(T)-1}.
\end{equation}
For the upper bound in~\eqref{eq:bounds-sketch}, let $J\subseteq Z$ be an $\mathcal{F}$-free subgraph with $\mathcal{N}(J,T)=\ex(Z, T, \mathcal{F})$.
Then by \eqref{eq:hom-upper-sketch} and \eqref{eq:aut-sketch}, we have
\begin{equation}\label{eq:upper-sketch}
\ex(Z, T, \mathcal{F}) = \mathcal{N}(J,T) \leq \gamma^{-1}B^{v(T)}\Big(A_0(\mathbf{x},G) +S(\mathbf{x},n)\cdot z +S(\mathbf{x},n)C\big(n^{k-\eta}B^{-\eta} +n^{-3}\big)\Big)
\end{equation}
for some constant $C>0$.
Set $L= \gamma^{-1}B^{v(T)}A_0(\mathbf{x},G)$, $S= \gamma^{-1}B^{v(T)}S(\mathbf{x},n)$, and let $a$ and $b$ be the error terms in \eqref{eq:lower-sketch} and \eqref{eq:upper-sketch}, respectively.
Then \eqref{eq:bounds-sketch} follows from \eqref{eq:lower-sketch} and \eqref{eq:upper-sketch}.
By choosing the construction parameters carefully, we can show that there exists a constant $N^{\ast}>0$ such that $a+b+v(Z)^{v(T)-\varepsilon}<S/2$ whenever $n\geq N^{\ast}$.
As we discussed above, this implies that each of the two approximation problems is NP-hard even when the input graph is restricted to $\{Z(G)\colon\, G\in\mathcal{G}_{k}, v(G)\geq N^{\ast}\}$, and hence also for arbitrary host graphs.

\medskip
\noindent\textbf{Further cases (Theorems~\ref{thm:multipar} and \ref{thm:chromatic+2}).}
In fact, the symmetric rooted construction below yields the more general Theorem~\ref{thm:chromatic}, which implies Theorem~\ref{thm:chromatic+2}, and also establishes Theorem~\ref{thm:clique}.
For these theorems, let $A$ and $B$ be positive integers and fix an oriented edge $xy\in E(T)$.
We construct $Z=Z_{A,B}(G,T,xy)$ in the following steps (see Figure~\ref{fig:root}):
\begin{enumerate}
\item[(1)] We begin with the empty graph on $V(G)$. For every ordered pair $(u,v)$ of distinct vertices of $G$, take a copy $T_{u,v}$ of the graph $T-\{x,y\}$, where all these copies are pairwise disjoint and are disjoint from $V(G)$.
\item[(2)] For each ordered pair $(u,v)$, replace $T_{u,v}$ by its $A$-blow-up $T_{u,v}(A)$.
\item[(3)] For each ordered pair $(u,v)$, regard $u$ and $v$ as playing the roles of $x$ and $y$, respectively, and add edges as follows: for each vertex $w\in V(T-\{x,y\})$,
    \begin{itemize}
        \item if $xw\in E(T)$, add all edges between $u$ and the cluster corresponding to $w$ in $T_{u,v}(A)$;
        \item if $yw\in E(T)$, add all edges between $v$ and the cluster corresponding to $w$ in $T_{u,v}(A)$.
    \end{itemize}
    No edge between $u$ and $v$ is added at this step.
\item[(4)] Add all edges of $G$ among the original vertices of $G$.
\item[(5)] Take the $B$-blow-up of the resulting graph.
\end{enumerate}
For Theorem~\ref{thm:multipar}, we use a related but more specialized construction, attaching complete multipartite auxiliary structures to edges of $G$ to obtain a sharp leading term in the number of $T$-copies.

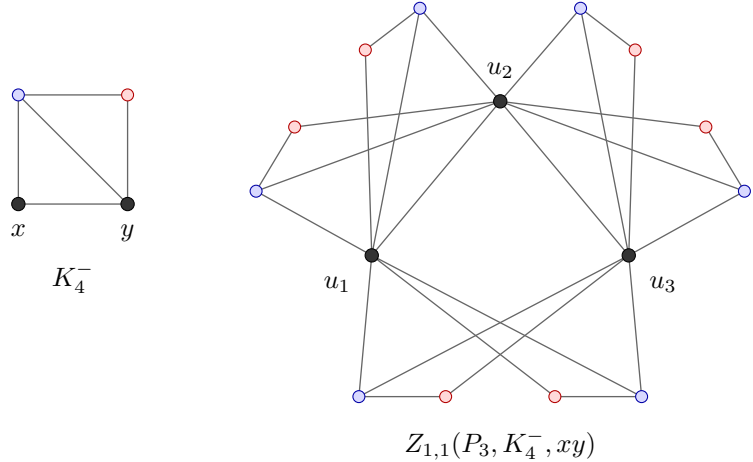
\begin{figure}[htbp]
\centering
\begin{tikzpicture}[
  x=.85cm,y=.85cm,
  core/.style={circle,draw=black,fill=black!80, inner sep=0pt,minimum size=5pt},
  rvertex/.style={circle,draw=blue!65!black,fill=blue!15, inner sep=0pt,minimum size=4.5pt},
  svertex/.style={circle,draw=red!70!black,fill=red!15, inner sep=0pt,minimum size=4.5pt},
  edge/.style={draw=black!60,line width=0.5pt},
  every node/.style={font=\small}
]
\coordinate (tx) at (-7.5,0.8); \coordinate (ty) at (-5.8,0.8); \coordinate (tr) at (-7.5,2.5); \coordinate (ts) at (-5.8,2.5);
\draw[edge] (tx)--(tr)--(ts)--(ty)--(tx); \draw[edge] (tr)--(ty);
\node[core] at (tx) {}; \node[core] at (ty) {};
\node[rvertex] at (tr) {}; \node[svertex] at (ts) {};
\node[below=4pt] at (tx) {$x$}; \node[below=4pt] at (ty) {$y$};
\node at (-6.65,-0.35) {$K_4^-$};

\coordinate (u1) at (-2,0); \coordinate (u2) at (0,2.4); \coordinate (u3) at (2,0);

\coordinate (r12) at (-3.8,1.0); \coordinate (s12) at (-3.2,2.0); \coordinate (r32) at (3.8,1.0); \coordinate (s32) at (3.2,2.0);

\coordinate (r21) at (-1.25,3.85); \coordinate (s21) at (-2.1,3.2); \coordinate (r23) at (1.25,3.85); \coordinate (s23) at (2.1,3.2);

\coordinate (r31) at (-2.2,-2.2); \coordinate (s31) at (-0.85,-2.2); \coordinate (r13) at (2.2,-2.2); \coordinate (s13) at (0.85,-2.2);

\draw[edge] (u1)--(r12)--(u2)--(s12)--(r12);
\draw[edge] (u3)--(r32)--(u2)--(s32)--(r32);
\draw[edge] (u2)--(r21)--(u1)--(s21)--(r21);
\draw[edge] (u2)--(r23)--(u3)--(s23)--(r23);
\draw[edge] (u3)--(r13)--(u1)--(s13)--(r13);
\draw[edge] (u1)--(r31)--(u3)--(s31)--(r31);
\draw[edge] (u1)--(u2)--(u3);

\foreach \p in {r12,r32,r21,r23,r31,r13}
  \node[rvertex] at (\p) {};
\foreach \p in {s12,s32,s21,s23,s31,s13}
  \node[svertex] at (\p) {};
\foreach \p in {u1,u2,u3}
  \node[core] at (\p) {};
\node[below left=4pt] at (u1) {$u_1$}; \node[above=5pt] at (u2) {$u_2$}; \node[below right=4pt] at (u3) {$u_3$};
\node at (0,-2.95) {$Z_{1,1}(P_3,K_4^-,xy)$};
\end{tikzpicture}
\caption{An illustration of symmetric rooted extensions with $A=B=1$, $G=P_3$ and $T=K_4^-$, where $K_4^-$ is obtained from $K_4$ by deleting one edge. For each ordered pair of distinct vertices of $P_3$, attach a copy of $T-\{x,y\}$ with the adjacencies prescribed by the construction.}
\label{fig:root}
\end{figure}

For the upper and lower bounds, we only present the differences from the proof for bipartite graphs.
Moreover, the following sketch mainly concerns Theorem~\ref{thm:chromatic+2}.
The proof of Theorem~\ref{thm:multipar} requires some suitable adjustments, and for simplicity, we omit the details.
In this proof, we estimate copies of $T$ directly, instead of estimating homomorphisms first and then passing to copies.

For the upper bound, we take an $\mathcal{F}$-free subgraph $J$ of $Z$, and define $J^{-}, J^{+}, R_J$ similarly.
We first show that $$\mathcal{N}(J,T) \leq ~\mathcal{N}(J^{+},T) \leq L+Xe(J^{-})+O_T\left(n^{2v(T)-3}A^{v(T)-3}B^{v(T)}\right)$$
for some numbers $L$ and $X$.
Then we use properties of $R_J$ to prove that $e(J^{-})\leq B^2z+cn^{k-\eta}B^{2-\eta}$ for some constants $c>0$ and $\eta>0$.
Combining these estimates, we have
\begin{equation}\label{eq:clique-upper-sketch}
\ex(Z, T, \mathcal{F})\leq L+XB^2z +O_{T, \mathcal{F}}\left(Xn^{k-\eta}B^{2-\eta} + n^{2v(T)-3}A^{v(T)-3}B^{v(T)}\right).
\end{equation}

For the lower bound, we choose a spanning subgraph $H\subseteq G$ with $\chi(H)\leq k-1$ and $e(H)=z$.
Let $Z_0$ be the spanning subgraph of $Z$ obtained by deleting all edges of $G(B)$, and $Z_0+H(B)$ be the graph obtained from $Z_0$ by adding edges of $H(B)$.
Then $Z_0+H(B)$ is a spanning subgraph of $Z$.
If we can show that $Z_0+H(B)$ is $\mathcal{F}$-free, then we have
\begin{equation}\label{eq:clique-lower-sketch}
\ex'(Z, T, \mathcal{F})\geq \mathcal{N}(Z_0+H(B),T)\geq L+Xe(H(B)) = L+XB^2z.
\end{equation}
Hence, it remains to prove the $\mathcal{F}$-freeness.
We use different methods to prove this for different theorems.
This is also where the assumptions on $T$ and $\mathcal{F}$ in the relevant theorem come from.
For Theorem~\ref{thm:chromatic+2}, we show that $\chi(Z_0+H(B))\leq k-1$ by extending a proper $(k-1)$-coloring of $H$ to a proper $(k-1)$-coloring of $Z_0+H(B)$, so $Z_0+H(B)$ is $\mathcal{F}$-free since $\min_{F\in \mathcal{F}}\chi(F)=k$.

Finally, set $S=XB^2$, $a=0$, and let $b$ be the error term in~\eqref{eq:clique-upper-sketch}.
Then~\eqref{eq:bounds-sketch} follows from~\eqref{eq:clique-upper-sketch} and~\eqref{eq:clique-lower-sketch}.
By choosing $A, B$ carefully so that $a+b+v(Z)^{v(T)-\varepsilon}<S/2$ for every sufficiently large $n$, we complete the proof.

\medskip
\noindent\textbf{Comparison of the constructions.}
Iterative leaf additions preserve the colorability needed to verify the $\mathcal{F}$-freeness of the lower-bound graph, while the counting part for both lower and upper bounds relies on the bipartiteness.
In the proof using symmetric rooted extensions, we can count copies of $T$ directly instead of counting homomorphisms first, but the feasibility of the lower-bound graph requires additional structural hypotheses on $T$ and $\mathcal{F}$.
The two approaches are complementary, each being applicable to different settings.

\section{Preliminaries}
\label{sec:pre}

We first establish a relation between $\ex(G, T, \mathcal{F})$ and $\ex'(G, T, \mathcal{F})$.

\begin{proposition}\label{prop:ex}
Let $T$ be a fixed graph with $e(T)>0$, and let $\mathcal{F}$ be a fixed nonempty family of graphs with $e(F)>0$ for every $F\in \mathcal{F}$.
Then the following statements hold.
\begin{itemize}
\item[{\rm (i)}] If $\mathcal{F}$ is finite, then $0\leq \ex(G, T, \mathcal{F})-\ex'(G, T, \mathcal{F})=O_{T,\mathcal{F}}(1)$ for every sufficiently large graph $G$.
\item[{\rm (ii)}] If no member of $\mathcal{F}$ has isolated vertices, then $\ex(G, T, \mathcal{F})=\ex'(G, T, \mathcal{F})$ for every graph $G$.
\end{itemize}
\end{proposition}

\noindent {\bf Proof.}
(i) By definition, we have $\ex(G, T, \mathcal{F})-\ex'(G, T, \mathcal{F})\geq 0.$
Set $M=\max_{F\in \mathcal{F}}v(F)$.
We show that $$\ex(G, T, \mathcal{F})-\ex'(G, T, \mathcal{F})\leq \mathcal{N}(K_{M-1}, T)$$ for every graph $G$, which implies the conclusion since $\mathcal{N}(K_{M-1}, T)$ depends only on $T$ and $\mathcal{F}$.
Let $H$ be an $\mathcal{F}$-free subgraph of $G$ with $\mathcal{N}(H,T)=\ex(G, T, \mathcal{F})$.
If $v(H)\leq M-1$, then $\ex(G, T, \mathcal{F})-\ex'(G, T, \mathcal{F})\leq \ex(G, T, \mathcal{F})=\mathcal{N}(H,T)\leq \mathcal{N}(K_{M-1}, T)$.
Now assume $v(H)\geq M$.
Let $H'$ be the spanning subgraph of $G$ obtained from $H$ by adding the missing vertices as isolated vertices.
We claim that $H'$ is also $\mathcal{F}$-free.
Indeed, if $H'$ contains a copy of some $F\in \mathcal{F}$, then every nonisolated vertex of $F$ must be mapped into $H$.
Since $v(H)\geq M\geq v(F)$, the isolated vertices of $F$ can also be mapped to the unused vertices of $H$.
This implies that $H$ contains a copy of $F$, a contradiction.
Thus $\ex'(G, T, \mathcal{F})\geq \mathcal{N}(H',T)\geq \mathcal{N}(H,T)=\ex(G, T, \mathcal{F})$, so $\ex(G, T, \mathcal{F})=\ex'(G, T, \mathcal{F})$ in this case.
The result follows.

(ii) It suffices to show that $\ex(G, T, \mathcal{F})\leq \ex'(G, T, \mathcal{F})$ if no member of $\mathcal{F}$ has isolated vertices.
Let $H$ be an $\mathcal{F}$-free subgraph of $G$ with $\mathcal{N}(H,T)=\ex(G, T, \mathcal{F})$.
Extend $H$ to a spanning subgraph $H'$ of $G$ by adding the missing vertices as isolated vertices.
Since no $F\in \mathcal{F}$ contains isolated vertices, $H'$ is also $\mathcal{F}$-free.
Thus $\ex'(G, T, \mathcal{F})\geq \mathcal{N}(H',T)\geq \mathcal{N}(H,T)=\ex(G, T, \mathcal{F})$.
This completes the proof.
\hfill$\blacksquare$

\vspace{0.2cm}

The following reduction result converts estimates for the extremal parameters into approximation hardness.

\begin{theorem}\label{thm:transfer}
Let $T$ be a fixed graph with $e(T)>0$, $\mathcal{F}$ be a nonempty family of graphs, and $\varepsilon>0$.
Let $\mathcal{G}$ be a class of input graphs,
and suppose that $z=z(G)$ is a graph parameter that is NP-hard to approximate within an additive error of $g(n)$ for $G\in \mathcal{G}$.
Suppose that for every $n$-vertex graph $G\in \mathcal{G}$, one can construct in polynomial time a graph $Z=Z(G)$ so that
$$L+Sz-a \leq \ex'(Z, T, \mathcal{F})\leq \ex(Z, T, \mathcal{F})\leq L+Sz+b,$$
where $L=L(G)\geq 0$ and $S=S(G)>0$ are polynomial-time computable rational numbers, and $a,b$ are nonnegative errors.
If $$a+b+v(Z)^{v(T)-\varepsilon}\leq g(n)S,$$ then it is NP-hard to approximate $\ex(Z, T, \mathcal{F})$ or $\ex'(Z, T, \mathcal{F})$ within an additive error of $v(Z)^{v(T)-\varepsilon}$ for $Z\in \{Z(G)\colon\, G\in\mathcal{G}\}$.
\end{theorem}

\noindent {\bf Proof.}
It suffices to consider $\ex(Z, T, \mathcal{F})$, since the proof for $\ex'(Z, T, \mathcal{F})$ is essentially the same.
Suppose that there is a polynomial-time algorithm which returns a number $x$ satisfying $|x-\ex(Z, T, \mathcal{F})|\leq v(Z)^{v(T)-\varepsilon}$.
Let $z'=\frac{x-L}{S}$.
Since $L+Sz-a \leq \ex(Z, T, \mathcal{F})\leq L+Sz+b$ and $a+b+v(Z)^{v(T)-\varepsilon}\leq g(n)S$, we have
\begin{align*}
  |z'-z| = &~\left|\frac{x-L}{S}-z\right| = \frac{1}{S}\left|x-(L+Sz)\right| \leq \frac{1}{S}\big(\left|x-\ex(Z, T, \mathcal{F})\right|+\left|\ex(Z, T, \mathcal{F})-(L+Sz)\right|\big)\\
   \leq &~\frac{1}{S}\left(v(Z)^{v(T)-\varepsilon}+a+b\right) \leq g(n),
\end{align*}
where the first inequality follows from the triangle inequality.
This gives a polynomial-time algorithm for approximating $z$ within an additive error of $g(n)$, although this approximation problem is assumed to be NP-hard.
Therefore, it is NP-hard to approximate $\ex(Z, T, \mathcal{F})$ within an additive error of $v(Z)^{v(T)-\varepsilon}$ for $Z\in \{Z(G)\colon\, G\in\mathcal{G}\}$.
The proof of Theorem~\ref{thm:transfer} is complete.
\hfill$\blacksquare$


\begin{corollary}\label{cor:transfer}
Let $T$ be a fixed graph with $e(T)>0$, $\mathcal{F}$ be a nonempty family of graphs with $k=\min_{F\in \mathcal{F}}\chi(F)\geq 3$, and $\varepsilon>0$.
Let $\mathcal{G}$ be a class of input graphs such that computing $z=\cut_{k-1}(G)$ is NP-hard for inputs $G\in \mathcal{G}$.
Suppose that for every $n$-vertex graph $G\in \mathcal{G}$, one can construct in polynomial time a graph $Z=Z(G)$ so that
$$L+Sz-a \leq \ex'(Z, T, \mathcal{F})\leq \ex(Z, T, \mathcal{F})\leq L+Sz+b,$$
where $L=L(G)\geq 0$ and $S=S(G)>0$ are polynomial-time computable rational numbers, and $a,b$ are nonnegative errors.
If $$a+b+v(Z)^{v(T)-\varepsilon}< \frac{S}{2},$$ then it is NP-hard to approximate $\ex(Z, T, \mathcal{F})$ or $\ex'(Z, T, \mathcal{F})$ within an additive error of $v(Z)^{v(T)-\varepsilon}$ for $Z\in \{Z(G)\colon\, G\in\mathcal{G}\}$.
\end{corollary}

\noindent {\bf Proof.}
For each $G\in \mathcal{G}$, define $g(G)\colonequals S^{-1}\left(a+b+v(Z)^{v(T)-\varepsilon}\right)$.
Let $\mathcal{G}(n)\colonequals \{G\in \mathcal{G}\colon\, v(G)=n\}$ and $g(n)\colonequals \max_{G\in \mathcal{G}(n)}g(G)$, with $g(n)\colonequals 0$ if $\mathcal{G}(n)=\varnothing$.
By the assumption, we have $0\leq g(n)<1/2$ for every $n$.
Since $z=\cut_{k-1}(G)$ is an integer, if we can approximate $z$ within an additive error of $g(n)<1/2$, then we can determine $z$ exactly by rounding to the nearest integer.
Hence, approximating $z$ within an additive error of $g(n)$ is as hard as determining $z$ exactly, and thus NP-hard.
Now the result follows from Theorem~\ref{thm:transfer}.
\hfill$\blacksquare$

\vspace{0.2cm}

We will use the following result of Alon, Shapira and Sudakov, which follows from \cite[Theorem~6.1~(i)]{AlSS} and the argument in \cite[Section~8]{AlSS}.
Lemma~\ref{le:AlSS-NPhard} below serves as a convenient reformulation of their results.
For completeness, we also supply a proof sketch of the lemma.

\begin{lemma}[\cite{AlSS}]\label{le:AlSS-NPhard}
For every integer $k\geq 3$, there exist constants $\mu=\mu(k)>0$ and $n_0=n_0(k)$ such that every $n$-vertex graph $G$ with $n\geq n_0$ and $\delta(G)\geq (1-\mu)n$ satisfies $\cut_{k-1}(G)=\ex(G, K_2, K_k)$.
Moreover, computing $\cut_{k-1}(G)$ exactly is NP-hard when restricted to these graphs.
\end{lemma}

\noindent {\bf Proof sketch.}
Applying \cite[Theorem~6.1(i)]{AlSS} with $H=K_k$, we obtain constants $\mu>0$ and $n_0$ such that $\cut_{k-1}(G)=\ex(G,K_2,K_k)$ whenever $v(G)\geq n_0$ and $\delta(G)\ge(1-\mu)v(G)$.
With the same choice of $\mu$, the reduction in \cite[Section~8]{AlSS} produces a graph $G$ containing a spanning regular graph $G_0$ of degree $(1-\mu)v(G)$, so $\delta(G)\ge(1-\mu)v(G)$.
Claim~8.3 and the subsequent rounding argument in that section show that computing the minimum number of edges whose deletion makes $G$ $(k-1)$-colorable is NP-hard even on these graphs.
Since this number equals $e(G)-\cut_{k-1}(G)$, the hardness assertion follows.
\hfill$\blacksquare$

\vspace{0.2cm}

For each integer $k\geq 3$, fix $\mu=\mu(k)>0$ and $n_0=n_0(k)$ as in the preceding lemma, and define $$\mathcal{G}_k\colonequals \big\{G\colon\, v(G)\geq n_0, \delta(G)\geq (1-\mu)v(G)\big\}.$$
Thus $\cut_{k-1}(G)$ remains NP-hard to compute on $\mathcal{G}_k$, while on this class it equals $\ex(G,K_2,K_k)$.

We will also use the following result of Alon and Shikhelman~\cite{AlSh16} on generalized Tur\'{a}n problems.

\begin{lemma}[{\cite[Proposition~2.1]{AlSh16}}]\label{le:AlSh16-degen}
Let $T$ and $F$ be two graphs with $v(T)>0$ and $v(F)>0$, and suppose that $F$ is a subgraph of a blow-up of $T$.
Then there exist constants $c=c(T,F)>0$ and $\eta=\eta(T,F)\in (0,v(T))$ such that every $n$-vertex $F$-free graph $H$ satisfies $\mathcal{N}(H,T)\leq cn^{v(T)-\eta}$.
\end{lemma}

\section{Proof of Theorem~\ref{thm:bipar}}
\label{sec:pf_bipar}

In this section, we confirm Conjecture~\ref{conj:AlSh22} for every bipartite graph $T$ with at least one edge.
We divide the proof into three steps.
In Section~\ref{subsec:leaf}, we introduce a construction $L_{\mathbf{x}}(G)$ based on iterative leaf additions, and derive an expansion for the number of homomorphisms from $T$ to $L_{\mathbf{x}}(G)$.
In Section~\ref{subsec:blow-up}, we construct a graph $Z$ from $L_{\mathbf{x}}(G)$ via a blow-up operation, and bound the number of homomorphisms from $T$ to its subgraphs using random transversals.
In Section~\ref{subsec:proof_bipar}, we complete the proof of Theorem~\ref{thm:bipar} by applying Corollary~\ref{cor:transfer}.

Throughout this section, we fix a bipartite graph $T$ with $e(T)>0$, a nonempty family $\mathcal{F}$ such that no $F\in \mathcal{F}$ is a subgraph of a blow-up of $T$, and $0<\varepsilon\le1$.
Let $G$ be an $n$-vertex graph, where $n\geq 1$.
If $v(T)=2$, then $T=K_2$, and the conclusion follows from the work of Alon, Shapira and Sudakov~\cite{AlSS} as explained in Section~\ref{sec:introduction}.
Henceforth, we assume that $v(T)\geq 3$ throughout this section.

Note that the condition on $\mathcal{F}$ is equivalent to $\min_{F\in \mathcal{F}}\chi(F)\geq 3$.
Indeed, if some $F\in \mathcal{F}$ satisfies $\chi(F)\leq 2$, then $F$ is a subgraph of a sufficiently large blow-up of an edge of $T$.
Conversely, if every $F\in \mathcal{F}$ satisfies $\chi(F)\geq 3$, then no $F\in \mathcal{F}$ is a subgraph of a blow-up of $T$ since every blow-up of a bipartite graph is bipartite.

\subsection{Iterative leaf additions}
\label{subsec:leaf}

In a graph, a {\it leaf} is a vertex of degree 1.
Given a graph $G$ and a positive integer $x$, let $L_x(G)$ be the graph obtained from $G$ by attaching $x$ leaves to every vertex of $G$, where all added vertices are distinct and lie outside $V(G)$.
See Figure~\ref{fig:leaf} for an illustration.
For each vertex $v\in V(G)$, we call $v$ the {\it anchor} of the $x$ leaves attached to $v$.
In order to study homomorphisms from a bipartite graph $Q$ to $L_x(G)$, we introduce the following {\it quotient graph} $Q_I$ of $Q$ with respect to an independent set $I$.
For an independent set $I\subseteq V(Q)$, we construct the graph $Q_I$ from $Q$ as follows:
\begin{itemize}
\item Delete the vertices of $I$.
\item For each nonisolated vertex $v\in I$, identify all vertices of $N_Q(v)$ to form a single vertex $v'$, taking the transitive closure of these identifications. The vertices $v'$ are not required to be distinct: if the transitive closure merges the identifications corresponding to two distinct vertices $v_1, v_2\in I$, then $v'_1$ and $v'_2$ denote the same vertex.
\item For each isolated vertex $v\in I$, add one new isolated vertex $v'$.
\end{itemize}

Note that if $e(Q_I)>0$, then every edge of $Q_I$ can be traced back to an edge of $Q$.
Indeed, deleting the vertices of $I$ and adding isolated vertices cannot create an edge.
For a fixed identification in which a set of vertices is identified as a single vertex $u$, an edge $uw$ is present in $Q_I$ only if $w$ was adjacent to at least one of the vertices identified as $u$;
all edges not incident to $u$ are inherited unchanged.
Thus every edge of $Q_I$ after an identification comes from an edge of $Q$ before that identification.
Therefore, tracing an edge of $Q_I$ backwards through all the identifications gives an edge of $Q$.
For an edge $e\in E(Q_I)$, we call the corresponding edge in $Q$ a {\it preimage edge} of $e$.
Notice that for each edge of $Q_I$, its preimage edge is not necessarily unique.
In particular, if $e(Q_I)>0$, then $e(Q)>0$.

We now provide an expansion for $\hom(Q, L_x(G))$ in terms of $\hom(Q_I, G)$.

\begin{lemma}\label{le:leaf-hom-1}
For any graph $G$, bipartite graph $Q$ and positive integer $x$, we have
$$\hom(Q, L_x(G))=\sum\limits_{\substack{I\subseteq V(Q),  \\ \text{\rm $I$ is independent}}}x^{|I|}\hom(Q_I, G).$$
\end{lemma}

\noindent {\bf Proof.}
Fix a homomorphism $\varphi$ from $Q$ to $L_x(G)$.
Let $I\subseteq V(Q)$ be the set of vertices mapped to the set $V(L_x(G))\setminus V(G)$ of newly added leaves.
Then $I$ is an independent set of $Q$ since $V(L_x(G))\setminus V(G)$ is an independent set of $L_x(G)$.
For any nonisolated vertex $v\in I$, all vertices of $N_Q(v)$ must map to the anchor of $\varphi(v)$.
This gives precisely the identifications in the definition of $Q_I$.
For any isolated vertex $v\in I$, the anchor of $\varphi(v)$ can be viewed as $\psi(v')$ for some homomorphism $\psi$ from $Q_I$ to $G$.
Once we fix a homomorphism $\psi$ from $Q_I$ to $G$, for each $v\in I$ we can choose $\varphi(v)$ independently from the $x$ leaves attached to $\psi(v')$.
This is a bijection, so the number of homomorphisms with this fixed set $I$ is $x^{|I|}\hom(Q_I,G)$.
Summing over independent sets $I$ proves the identity.
\hfill$\blacksquare$

\vspace{0.2cm}

The following properties of $Q_I$ will be used in our proof of Lemma~\ref{le:leaf-hom-2} below.

\begin{lemma}\label{le:QI}
Let $Q$ be a bipartite graph and $I$ be an independent set of $Q$.
Then
\begin{itemize}
\item[{\rm (i)}] $v(Q_I)\leq v(Q)$,
\item[{\rm (ii)}] $Q_I$ is bipartite.
\end{itemize}
\end{lemma}

\noindent {\bf Proof.}
(i) In the construction of $Q_I$, each isolated vertex of $I$ is replaced by one isolated vertex, while every nonisolated vertex of $I$ is deleted.
Let $I'$ be the set of nonisolated vertices of $I$.
Then $v(Q_I)\leq v(Q)-|I'|\leq v(Q)$.

(ii) Let $A, B$ be the partite sets of $Q$.
For any nonisolated vertex $v\in I\cap A$, we have $N_Q(v)\subseteq B$, so after identifying $N_Q(v)$ we still obtain a bipartite graph.
Similarly, for any nonisolated vertex $v\in I\cap B$, we have $N_Q(v)\subseteq A$, so after identifying $N_Q(v)$ we still obtain a bipartite graph.
Therefore, even after taking the transitive closure of all these identifications, the resulting graph is bipartite.
Moreover, the new vertices corresponding to isolated vertices can be placed in either part.
Hence, $Q_I$ is a bipartite graph.
\hfill$\blacksquare$

\vspace{0.2cm}

Let $s\colonequals v(T)-2$ and $\mathbf{x}=(x_1, \ldots, x_s)$ be a vector of $s$ positive integers, so $\mathbf{x} \in \mathbb{Z}_{>0}^s$.
Note that $s\geq 1$ since $v(T)\geq 3$.
Define $$L_{\mathbf{x}}(G)\colonequals L_{x_1}\left(L_{x_2}\left(\cdots L_{x_s}\left(G\right) \cdots\right)\right).$$
We define a sequence $(I_1,\ldots,I_s)$ of vertex sets and a sequence $(T_0, T_1, \ldots, T_s)$ of graphs as follows:
\begin{equation}\label{eq:Ti}
T_0\colonequals T, \qquad
I_i\text{ is an independent set of }T_{i-1},\qquad
T_i=(T_{i-1})_{I_i}
\end{equation}
for every $i=1, \ldots, s$.

Applying Lemma~\ref{le:leaf-hom-1} iteratively, we obtain the following expansion result.

\begin{lemma}\label{le:leaf-hom-2}
For any graph $G$, bipartite graph $T$, and vector $\mathbf{x}=(x_1, \ldots, x_s) \in \mathbb{Z}_{>0}^s$, we have
$$\hom(T,L_{\mathbf{x}}(G)) = \sum_{\text{\rm $(I_1,\ldots,I_s)$ satisfying~\eqref{eq:Ti}}}\left(\prod_{i=1}^s x_i^{|I_i|}\right)\hom(T_s,G).$$
\end{lemma}

\noindent {\bf Proof.}
By Lemma~\ref{le:QI}, each $T_i$ defined in \eqref{eq:Ti} is bipartite and $v(T_s)\leq v(T_{s-1})\leq \cdots \leq v(T_0)=v(T)$.
Set $G_{s+1}=G$, and define $G_j=L_{x_j}(G_{j+1})$ successively for $j=s,s-1,\ldots,1$.
In particular, $G_1=L_{\mathbf{x}}(G)$.
Applying Lemma~\ref{le:leaf-hom-1}, we have
$$\hom(T,L_{\mathbf{x}}(G)) = \sum_{\substack{I_1\subseteq V(T_0), \\ I_1\text{ is independent}}}x_1^{|I_1|}\hom(T_1, G_2).$$
When $s\geq 2$, applying Lemma~\ref{le:leaf-hom-1} again to every term on the right-hand side, we have
$$\begin{aligned}
\hom(T,L_{\mathbf{x}}(G))= &~\sum_{\substack{I_1\subseteq V(T_0),\\ I_1\text{ is independent}}}x_1^{|I_1|} \sum_{\substack{I_2\subseteq V(T_1),\\  I_2\text{ is independent}}} x_2^{|I_2|} \hom(T_2,
G_3) \\
= &~\sum_{\substack{I_1\subseteq V(T_0),\\ I_1\text{ is independent}}} \sum_{\substack{I_2\subseteq V(T_1),\\  I_2\text{ is independent}}} x_1^{|I_1|}x_2^{|I_2|} \hom(T_2, G_3).
\end{aligned}$$
Continuing in this way, after $s$ applications of Lemma~\ref{le:leaf-hom-1}, we obtain
$$\hom(T,L_{\mathbf{x}}(G)) = \sum_{\text{\rm $(I_1,\ldots,I_s)$ satisfying~\eqref{eq:Ti}}}\left(\prod_{i=1}^s x_i^{|I_i|}\right)\hom(T_s,G).$$
The proof of Lemma~\ref{le:leaf-hom-2} is complete.
\hfill$\blacksquare$

\vspace{0.2cm}

For a graph $Q$, let $p(Q)\colonequals \left|\{v\in V(Q)\colon\, d_Q(v)>0\}\right|$ be the number of nonisolated vertices of $Q$.
The following properties of $p(Q)$ will be used in our proof of Lemma~\ref{le:K2} below.

\begin{lemma}\label{le:pQ}
Let $Q$ be a bipartite graph.
\begin{itemize}
\item[{\rm (i)}] For any independent set $I$, we have $p(Q_I)\leq p(Q)$.
\item[{\rm (ii)}] Suppose that $p(Q)\geq3$. If $I$ is an independent set of maximum cardinality subject to $e(Q_I)>0$, then $p(Q_I)<p(Q)$.
\end{itemize}
\end{lemma}

\noindent {\bf Proof.}
(i) In the construction of $Q_I$, all nonisolated vertices of $I$ (if any) are deleted.
At each identification step, if the resulting vertex $v'$ is nonisolated, then at least one of the vertices being identified was nonisolated immediately before that step.
Thus identification cannot increase the number of nonisolated vertices.
Moreover, adding isolated vertices does not affect this number.
Hence, we have $p(Q_I)\leq p(Q)$.

(ii) Assume that $p(Q)\geq3$.
Since $e(Q)>0$ and $Q_{\varnothing}=Q$, at least one independent set satisfies $e(Q_I)>0$.
Choose such a set $I$ of maximum cardinality.
We claim that $I$ contains a nonisolated vertex.
Suppose for a contradiction that all vertices in $I$ are isolated vertices.
Then all edges of $Q$ are contained in $Q-I$.
Choose an edge $uv\in E(Q)$.
Since $p(Q)\ge3$, there exists a nonisolated vertex $w\notin\{u,v\}$.
Since every vertex of $I$ is isolated, we have $w\notin I$, and $I\cup\{w\}$ is independent.
Note that $uv\in E(Q-(I\cup \{w\}))$.
Since $Q$ is bipartite, we have $|N_Q(w)\cap \{u,v\}|\leq 1$, so the image of $uv$ remains an edge of $Q_{I\cup\{w\}}$.
Then $e(Q_{I\cup \{w\}})>0$, and this contradicts the maximality of $I$.
Hence, $I$ contains a nonisolated vertex.
Since all nonisolated vertices of $I$ are deleted in the construction of $Q_I$, we have $p(Q_I)<p(Q)$.
\hfill$\blacksquare$

\vspace{0.2cm}

For any sequence $(I_1,\ldots,I_s)$ defined by \eqref{eq:Ti}, we call $(|I_1|,\ldots,|I_s|)$ the {\it exponent vector} of the monomial $\prod_{i=1}^s x_i^{|I_i|}$.
We order exponent vectors lexicographically, beginning with the first coordinate. That is, $(a_1,\ldots,a_s)>(b_1,\ldots,b_s)$ if for the smallest $i$ with $a_i\neq b_i$, we have $a_i>b_i$.
\vspace{0.2cm}

\begin{lemma}\label{le:K2}
Let $T$ be a bipartite graph with $e(T)>0$.
Among all sequences $(I_1,\ldots,I_s)$ defined by \eqref{eq:Ti} for which $e(T_s)>0$, choose one with the lexicographically largest exponent vector.
Then $p(T_s)=2$, and consequently, $T_s\cong K_2\cup qK_1$ for some integer $q\geq 0$.
\end{lemma}

\noindent {\bf Proof.}
First, there exists at least one sequence $(I_1,\ldots,I_s)$ satisfying $e(T_s)>0$.
Indeed, we may take $I_1=\cdots=I_s=\varnothing$.
Since $Q_{\varnothing}=Q$ for every graph $Q$, we have $T_0=T_1=\cdots=T_s=T$, and hence $e(T_s)=e(T)>0$.
Since there are only finitely many possible sequences, there exists one with the lexicographically largest exponent vector.
Fix such a sequence $(I_1,\ldots,I_s)$.
Since $e(T_s)>0$, we have $e(T_i)>0$ for each $i\in \{0, \ldots, s\}$.

\begin{claim}\label{cl:K2}
For each $i\in \{1, \ldots, s\}$, the set $I_i$ has maximum cardinality among all independent sets $I\subseteq V(T_{i-1})$ for which $e((T_{i-1})_I)>0$.
\end{claim}

\begin{proof}
For a contradiction, suppose that for some $j$ there is an independent set $J\subseteq V(T_{j-1})$ such that $e((T_{j-1})_J)>0$ and $|J|>|I_j|$.
Let $I'_i\colonequals I_i$ for $1\leq i\leq j-1$, $I'_j\colonequals J$, and $I'_i\colonequals \varnothing$ for $j+1\leq i\leq s$.
Write $T_0', \ldots, T_s'$ for the graphs produced by $(I_1', \ldots, I_s')$.
Then $e(T'_s)=e(T'_j)=e((T'_{j-1})_J)=e((T_{j-1})_J)>0$.
However, the exponent vector of $(I'_1,\ldots,I'_s)$ is lexicographically larger than that of $(I_1,\ldots,I_s)$, a contradiction.
\end{proof}

We now show that $p(T_s)=2$.
Since $e(T_s)>0$, we have $p(T_s)\geq 2$.
Suppose for a contradiction that $p(T_s)\geq 3$.
Then $p(T_0)\geq p(T_1)\geq \cdots \geq p(T_s)\geq 3$ by Lemma~\ref{le:pQ}~(i).
By Lemma~\ref{le:pQ}~(ii) and Claim~\ref{cl:K2}, we have $p(T_i)<p(T_{i-1})$ whenever $p(T_{i-1})\geq 3$.
Since these quantities are integers, we have $p(T_i)\leq p(T_{i-1})-1$ whenever $p(T_{i-1})\geq 3$.
Then $p(T_s)\leq p(T_0)-s \leq v(T)-(v(T)-2)=2$, a contradiction.
Hence, we have $p(T_s)=2$.
Since $T_s$ is a simple graph with exactly two nonisolated vertices and $e(T_s)>0$, these two vertices are adjacent and all other vertices are isolated.
Therefore, $T_s\cong K_2\cup qK_1$ for some integer $q\geq 0$.
The proof of Lemma~\ref{le:K2} is complete.
\hfill$\blacksquare$

\vspace{0.2cm}

By Lemma~\ref{le:QI}, every terminal graph $T_s$ in the expansion of Lemma~\ref{le:leaf-hom-2} is bipartite and has at most $v(T)$ vertices.
Let $\mathcal{Q}_T$ consist of exactly one representative of each isomorphism class of bipartite graphs on at most $v(T)$ vertices.
Thus $\mathcal{Q}_T$ is finite and depends only on $T$.
For each $Q\in \mathcal{Q}_T$, define
$$f_Q(\mathbf{x})\colonequals \sum_{\substack{\text{\rm $(I_1,\ldots,I_s)$ satisfying~\eqref{eq:Ti},}\\ T_s\cong Q}}\left(\prod_{i=1}^s x_i^{|I_i|}\right),$$
where the sum is taken over all sequences satisfying~\eqref{eq:Ti} with terminal graph $T_s$ isomorphic to $Q$.
We define $f_Q(\mathbf{x})\colonequals 0$ if no such sequence exists.
Then $f_Q(\mathbf{x})$ is a polynomial in $x_1, \ldots, x_s$ with nonnegative integer coefficients.
Grouping the terms in Lemma~\ref{le:leaf-hom-2} according to the isomorphism type of $T_s$, we obtain
\begin{equation*}
\hom(T, L_{\mathbf{x}}(G))=\sum\limits_{Q\in \mathcal{Q}_T}f_Q(\mathbf{x})\hom(Q, G).
\end{equation*}
Since no graph has exactly one nonisolated vertex, we have either $p(Q)=0$ or $p(Q)\geq 2$.
By the proofs of Lemmas~\ref{le:leaf-hom-1} and~\ref{le:leaf-hom-2}, each homomorphism from $T$ to $L_{\mathbf{x}}(G)$ determines a unique sequence $(I_1,\ldots,I_s)$ satisfying~\eqref{eq:Ti}.
We partition these homomorphisms according to whether the corresponding terminal graph $T_s$ satisfies $p(T_s)=0$, $p(T_s)=2$, or $p(T_s)\geq 3$:
\begin{equation}\label{eq:hom-3part}
\hom(T, L_{\mathbf{x}}(G))=A_0(\mathbf{x},G)+A_2(\mathbf{x},G)+A_{\geq 3}(\mathbf{x},G),
\end{equation}
where
\begin{align*}
A_0(\mathbf{x},G)\colonequals &~\sum\limits_{Q\in \mathcal{Q}_T,~p(Q)=0}f_Q(\mathbf{x})\hom(Q, G), \\
A_2(\mathbf{x},G)\colonequals &~\sum\limits_{Q\in \mathcal{Q}_T,~p(Q)=2}f_Q(\mathbf{x})\hom(Q, G), \\
A_{\geq 3}(\mathbf{x},G)\colonequals &~\sum\limits_{Q\in \mathcal{Q}_T,~p(Q)\geq 3}f_Q(\mathbf{x})\hom(Q, G).
\end{align*}
Let $n=v(G)$, and define
$$S(\mathbf{x},n) \colonequals \sum\limits_{Q\in \mathcal{Q}_T,~p(Q)=2} 2\cdot f_Q(\mathbf{x})n^{v(Q)-2}.$$
Define
\begin{equation}\label{eq:leaves}
d_s\colonequals v(T)+3, \qquad d_i\colonequals v(T)+3 + v(T)\cdot \sum\nolimits_{j=i+1}^{s} d_j, \qquad x_i\colonequals n^{d_i}
\end{equation}
for every $i=1, \ldots, s$.
The weights $d_i$ are chosen so that if $(|I^{\ast}_1|,\ldots,|I^{\ast}_s|)$ is lexicographically larger than $(|I_1|,\ldots,|I_s|)$, then $\sum_{i=1}^{s}d_i|I^{\ast}_i|> \sum_{i=1}^{s}d_i|I_i|$; see the proof of Lemma~\ref{le:leaf-hom-3} below.

\begin{lemma}\label{le:leaf-hom-3}
For any graph $G$ and bipartite graph $T$, if $\mathbf{x}=(x_1, \ldots, x_s)$ is defined as in \eqref{eq:leaves}, then
$$A_0(\mathbf{x},G)=\sum\limits_{Q\in \mathcal{Q}_T,~p(Q)=0}f_Q(\mathbf{x})n^{v(Q)}, \quad A_2(\mathbf{x},G)=S(\mathbf{x},n)\cdot e(G), \quad A_{\geq 3}(\mathbf{x},G)=S(\mathbf{x},n)\cdot O_T\left(n^{-3}\right).$$
\end{lemma}

\noindent {\bf Proof.}
If $p(Q)=0$, then $e(Q)=0$ and thus $\hom(Q, G)=v(G)^{v(Q)}=n^{v(Q)}$.
Hence, $$A_0(\mathbf{x},G)=\sum\limits_{Q\in \mathcal{Q}_T,~p(Q)=0}f_Q(\mathbf{x})n^{v(Q)}.$$

If $p(Q)=2$, then $Q$ is isomorphic to $K_2\cup qK_1$ for some $q\geq 0$.
Then $\hom(Q, G)=\hom(K_2\cup qK_1,G)=2e(G)n^{q}=2e(G)n^{v(Q)-2}.$
Hence, $A_2(\mathbf{x},G)=S(\mathbf{x},n)\cdot e(G).$

We now consider the case $p(Q)\geq 3$.
Among all sequences satisfying \eqref{eq:Ti} whose terminal graph contains an edge, fix one sequence $(I^{\ast}_1,\ldots,I^{\ast}_s)$ with the lexicographically largest exponent vector.
Let $T^{\ast}_s$ be its terminal graph and let $Q'\in \mathcal{Q}_T$ satisfy $T^{\ast}_s\cong Q'$.
By Lemma~\ref{le:K2}, we have $p(Q')=2$.
The chosen sequence contributes
$\prod_{i=1}^s x_i^{|I^{\ast}_i|}=\prod_{i=1}^s n^{d_i|I^{\ast}_i|}=n^{\sum\nolimits_{i=1}^{s}d_i|I^{\ast}_i|}$ to $f_{Q'}(\mathbf{x})$, so
\begin{equation}\label{eq:leaf-hom-3-1}
f_{Q'}(\mathbf{x})\geq n^{\sum_{i=1}^s d_i|I^{\ast}_i|}.
\end{equation}
Consider any sequence $(I_1,\ldots,I_s)$ with $p(T_s)\geq 3$.
Its exponent vector is strictly smaller than $(|I^{\ast}_1|, \ldots, |I^{\ast}_s|)$; otherwise it would also be maximal, and then $p(T_s)=2$ by Lemma~\ref{le:K2}, a contradiction.
Let $i_0$ be the smallest index with $|I_{i_0}|\neq |I^{\ast}_{i_0}|$.
Then $|I_{i_0}|\leq |I^{\ast}_{i_0}|-1$ and $|I_i|=|I^{\ast}_i|$ for $1\leq i\leq i_0-1$.
By \eqref{eq:leaves} and since $0\leq |I_i|,|I^{\ast}_i|\leq v(T)$, we have
\begin{align*}
\sum\nolimits_{i=1}^{s}d_i|I_i|-\sum\nolimits_{i=1}^{s}d_i|I^{\ast}_i|= &~d_{i_0}(|I_{i_0}|- |I^{\ast}_{i_0}|)+\sum\nolimits_{i=i_0+1}^{s}d_i(|I_{i}|- |I^{\ast}_{i}|) \\
\leq &~-d_{i_0}+\sum\nolimits_{i=i_0+1}^{s}d_iv(T) =-(v(T)+3).
\end{align*}
Hence, $n^{\sum\nolimits_{i=1}^{s}d_i|I_i|}\leq n^{\sum\nolimits_{i=1}^{s}d_i|I^{\ast}_i|}\cdot n^{-(v(T)+3)}.$
Combining this bound with \eqref{eq:leaf-hom-3-1} and using the fact that the number of sequences depends only on $T$, we obtain that for every $Q\in \mathcal{Q}_T$ with $p(Q)\geq 3$,
$$f_{Q}(\mathbf{x})=\sum_{\substack{\text{\rm $(I_1,\ldots,I_s)$ satisfying~\eqref{eq:Ti},}\\ T_s\cong Q}}n^{\sum\nolimits_{i=1}^{s}d_i|I_i|} =O_T\left(n^{\sum\nolimits_{i=1}^{s}d_i|I^{\ast}_i|}\cdot n^{-(v(T)+3)}\right) \leq O_T\left(f_{Q'}(\mathbf{x})\cdot n^{-(v(T)+3)}\right).$$
Since $\hom(Q,G)\leq n^{v(Q)}\leq n^{v(T)}$ and $|\mathcal{Q}_T|$ depends only on $T$, it follows that
\begin{align*}
A_{\geq 3}(\mathbf{x},G)= &~\sum\limits_{Q\in \mathcal{Q}_T,~p(Q)\geq 3}f_Q(\mathbf{x})\hom(Q, G) \leq \sum\limits_{Q\in \mathcal{Q}_T,~p(Q)\geq 3}f_Q(\mathbf{x})n^{v(T)} \\
\leq &~O_T\left(f_{Q'}(\mathbf{x})\cdot n^{-(v(T)+3)}\right)n^{v(T)} = f_{Q'}(\mathbf{x})\cdot O_T\left(n^{-3}\right) \\
\leq &~\sum\limits_{Q\in \mathcal{Q}_T,~p(Q)=2}2\cdot f_Q(\mathbf{x})n^{v(Q)-2}\cdot O_T\left(n^{-3}\right) = S(\mathbf{x},n)\cdot O_T\left(n^{-3}\right).
\end{align*}
This completes the proof of Lemma~\ref{le:leaf-hom-3}.
\hfill$\blacksquare$

\vspace{0.2cm}

Finally, we show that $S(\mathbf{x},n)$ is an integer greater than $1$.
Indeed, by the definition of $S(\mathbf{x},n)$, we know that it is an integer.
Moreover, Lemma~\ref{le:K2} supplies a sequence $(I_1,\ldots,I_s)$ with $p(T_s)=2$, so at least one $f_Q(\mathbf{x})$ in the definition of $S(\mathbf{x},n)$ is at least $1$.
Since $n\geq 1$ and $v(Q)\geq p(Q)\geq 2$, we have $S(\mathbf{x},n) = \sum\limits_{Q\in \mathcal{Q}_T,~p(Q)=2} 2\cdot f_Q(\mathbf{x})n^{v(Q)-2}\geq 2$.

\subsection{Blow-up operation and the random-transversal argument}
\label{subsec:blow-up}

In Section~\ref{subsec:leaf}, we constructed a graph $L_{\mathbf{x}}(G)$ from $G$ by iterative leaf additions, where $\mathbf{x}=(x_1, \ldots, x_s)$ is given by \eqref{eq:leaves}.
Let $Z(G)\colonequals L_{\mathbf{x}}(G)(B)$ be the $B$-blow-up of $L_{\mathbf{x}}(G)$, where $B$ is a positive integer to be chosen in Section~\ref{subsec:proof_bipar}.
We will simply write $Z$ for $Z(G)$.
Then
\begin{equation}\label{eq:vZ}
v(Z)=v(L_{\mathbf{x}}(G))B= \left(v(G)\cdot \prod\nolimits_{i=1}^{s}\left(x_i+1\right)\right)B=Bn\prod\nolimits_{i=1}^{s}\left(n^{d_i}+1\right)\leq 2^sBn^{1+\sum\nolimits_{i=1}^{s}d_i}.
\end{equation}
The clusters corresponding to $V(G)$ induce a copy of $G(B)$, which we identify with $G(B)$.
For any subgraph $J\subseteq Z$, define a graph $J^{+}$ by
$$V(J^{+})=V(Z) \quad\text{and}\quad E(J^{+})=\big(E(J)\cap E(G(B))\big) \cup\big(E(Z)\setminus E(G(B))\big).$$
Then $J\subseteq J^{+}\subseteq Z$, and thus $\hom(T, J)\leq \hom(T, J^{+})\leq \hom(T, Z)$.

Recall that for each $v\in V(L_{\mathbf{x}}(G))$, $C_v$ is the cluster corresponding to $v$.
Given a subgraph $J\subseteq Z$, choose one vertex $v'\in C_v$ independently and uniformly for each $v\in V(G)$.
The selected vertices form a random transversal $\{v':v\in V(G)\}$ of the clusters corresponding to $V(G)$.
Define a random spanning subgraph $R_J$ of $G$ by
$$V(R_J)\colonequals V(G), \qquad E(R_J)\colonequals \{uv\in E(G)\colon\, u',v'\in V(J),\ u'v'\in E(J)\}.$$
Thus $R_J$ is obtained by retaining the edges of $J$ between the selected vertices and relabelling each selected vertex $v'$ as $v$.
In particular, if $v'\notin V(J)$, then $v$ is isolated in $R_J$.

In the remainder of this section, the vector $\mathbf{x}=(x_1, \ldots, x_s)$ is defined by \eqref{eq:leaves}.
We now provide an upper bound on $\hom(T,J)$ for any subgraph $J\subseteq Z$ in terms of $\mathbb{E}[e(R_J)]$ and the parameters $A_0(\mathbf{x},G)$, $S(\mathbf{x},n)$ defined in Section~\ref{subsec:leaf}.

\begin{lemma}\label{le:leaf-hom-4}
For any graph $J\subseteq Z$, we have
$$\hom(T,J)\leq B^{v(T)}\Big(A_0(\mathbf{x},G)+S(\mathbf{x},n)\cdot \mathbb{E}[e(R_J)]+S(\mathbf{x},n)\cdot O_T\left(n^{-3}\right)\Big).$$
\end{lemma}

\noindent {\bf Proof.}
Since $J\subseteq J^{+}$, we have $\hom(T, J)\leq \hom(T, J^{+})$.
Hence, it suffices to prove that $\hom(T, J^{+})\leq B^{v(T)}\Big(A_0(\mathbf{x},G)+S(\mathbf{x},n)\cdot \mathbb{E}[e(R_J)]+S(\mathbf{x},n)\cdot O_T\left(n^{-3}\right)\Big).$

Let $\pi: V(Z)\to V(L_{\mathbf{x}}(G))$ be the natural projection that maps every vertex of a cluster to the corresponding vertex of $L_{\mathbf{x}}(G)$.
Since $Z$ is a $B$-blow-up of $L_{\mathbf{x}}(G)$, every homomorphism $\varphi: T\to J^{+}$ induces a homomorphism $\pi\circ\varphi: T\to L_{\mathbf{x}}(G)$.
Conversely, every homomorphism $\psi: T\to L_{\mathbf{x}}(G)$ can be lifted to exactly $B^{v(T)}$ homomorphisms from $T$ to $Z$,
since the image of each vertex of $T$ may be chosen independently in the cluster prescribed by $\psi$.
Note that not all of these homomorphisms are necessarily from $T$ to $J^{+}$.
However, since $J^{+}\subseteq Z$, we have $\hom(T, J^{+})\leq \hom(T, Z)$.
Since we only need an upper bound on $\hom(T, J^{+})$, we will count homomorphisms from $T$ to $Z$ when we analyze the $A_0(\mathbf{x},G)$-part and the $A_{\geq 3}(\mathbf{x},G)$-part below.

In the following, we count homomorphisms from $T$ to $J^{+}$ according to the three parts in~\eqref{eq:hom-3part}:
$$\hom(T, L_{\mathbf{x}}(G))=A_0(\mathbf{x},G)+A_2(\mathbf{x},G)+A_{\geq 3}(\mathbf{x},G).$$
For the $A_0(\mathbf{x},G)$-part, recall that every homomorphism $\psi: T\to L_{\mathbf{x}}(G)$ can be lifted to $B^{v(T)}$ homomorphisms from $T$ to $Z$, and hence at most $B^{v(T)}$ homomorphisms from $T$ to $J^{+}$.
Therefore, the contribution to $\hom(T,J^{+})$ from the $A_0(\mathbf{x},G)$-part is at most
\begin{equation}\label{eq:A0}
B^{v(T)}A_0(\mathbf{x},G).
\end{equation}
Similarly, the contribution to $\hom(T,J^{+})$ from the $A_{\geq 3}(\mathbf{x},G)$-part is at most $B^{v(T)}A_{\geq 3}(\mathbf{x},G).$
By Lemma~\ref{le:leaf-hom-3}, we have $A_{\geq 3}(\mathbf{x},G)=S(\mathbf{x},n)\cdot O_T\left(n^{-3}\right)$.
Hence, the contribution is at most
\begin{equation}\label{eq:A3}
B^{v(T)}S(\mathbf{x},n)\cdot O_T\left(n^{-3}\right).
\end{equation}

It remains to consider the $A_2(\mathbf{x},G)$-part.
For every edge $uv\in E(G)$, let $C_u$ and $C_v$ be the corresponding clusters in $G(B)$, and define
$d_J(u,v)\colonequals \frac{e_J(C_u,C_v)}{B^2},$
where $e_J(C_u,C_v)$ denotes the number of edges of $J$ with one endpoint in $C_u\cap V(J)$ and the other in $C_v\cap V(J)$.

\begin{claim}\label{cl:A2}
$\mathbb{E}[e(R_J)]= \sum_{uv\in E(G)}d_J(u,v).$
\end{claim}

\begin{proof}
For each $uv\in E(G)$, let $\mathbf{1}_{uv}$ be the indicator random variable of the event $uv\in E(R_J)$.
In the construction of $R_J$, one vertex is chosen independently and uniformly from each of $C_u$ and $C_v$.
Thus there are $B^2$ equally likely pairs of chosen vertices, and exactly $e_J(C_u,C_v)$ of these pairs form an edge of $J$.
Hence, $\mathbb{E}[\mathbf{1}_{uv}]=d_J(u,v)$.
Since $e(R_J)=\sum_{uv\in E(G)}\mathbf{1}_{uv}$, we have $\mathbb{E}[e(R_J)]= \sum_{uv\in E(G)}d_J(u,v)$ by linearity of expectation.
\end{proof}

Recall that
\begin{align*}
A_{2}(\mathbf{x},G) = &~\sum\limits_{Q\in \mathcal{Q}_T,~p(Q)=2}f_Q(\mathbf{x})\hom(Q, G) \\
= &~\sum\limits_{Q\in \mathcal{Q}_T,~p(Q)=2}~~\sum_{\substack{\text{\rm $(I_1,\ldots,I_s)$ satisfying~\eqref{eq:Ti},}\\ T_s\cong Q}}\left(\prod_{i=1}^s x_i^{|I_i|}\right)\hom(Q, G).
\end{align*}
Fix $Q\in \mathcal{Q}_T$ with $p(Q)=2$.
Then $Q\cong K_2\cup qK_1$ for some $q\geq0$.
Let $y_1y_2$ be the unique edge of $Q$.
Fix a sequence $(I_1,\ldots,I_s)$ satisfying~\eqref{eq:Ti} for which $T_s\cong Q$, and identify $T_s$ with $Q$.
We also fix a homomorphism $\theta: Q\to G$, and let $u= \theta(y_1)$ and $v= \theta(y_2)$.
Then $uv\in E(G)$.

Note that the edge $y_1y_2$ of $Q=T_s$ can be traced back to an edge $w_1w_2$ of $T$.
Indeed, since $T_i=(T_{i-1})_{I_i}$, every edge of $T_i$ has a preimage edge in $T_{i-1}$.
Starting with $y_1y_2\in E(T_s)$ and tracing backwards through $T_{s-1},\ldots,T_0$, we obtain an edge $w_1w_2\in E(T_0)=E(T)$ whose image after all the quotient operations is $y_1y_2$.

Fix one of the $\prod_{i=1}^s x_i^{|I_i|}$ homomorphisms $\psi: T\to L_{\mathbf{x}}(G)$ corresponding to the sequence $(I_1,\ldots,I_s)$ and the homomorphism $\theta$.
Since $u= \theta(y_1)$, $v= \theta(y_2)$ and the vertices $w_1, w_2$ are traced to $y_1, y_2$, respectively, we have $\psi(w_1)=u$ and $\psi(w_2)=v$,
where $u$ and $v$ are regarded as the corresponding vertices of the original copy of $G$ inside $L_{\mathbf{x}}(G)$.
In a homomorphism $\varphi: T\to J^{+}$ lifted from $\psi$, the images of $w_1$ and $w_2$ must form an edge of $J$ between $C_u$ and $C_v$.
There are at most $e_J(C_u,C_v)$ choices for these two images, and at most $B^{v(T)-2}$ choices for the remaining vertices.
Thus the number of such homomorphisms $\varphi$ lifted from $\psi$ is at most $e_J(C_u,C_v)B^{v(T)-2}=B^{v(T)}d_J(u,v).$

We now sum this bound over all homomorphisms $\theta: Q\to G$.
For each edge $uv\in E(G)$, the edge $y_1y_2$ of $Q$ can be mapped onto $uv$ in two ways, according to its orientation.
Moreover, each of the $q=v(Q)-2$ isolated vertices of $Q$ can be mapped independently to any of the $n$ vertices of $G$.
Therefore, we obtain
\begin{align*}
~ &~\sum_{\theta:Q\to G}e_J(C_{\theta(y_1)},C_{\theta(y_2)})B^{v(T)-2} = \sum_{\theta:Q\to G}B^{v(T)}d_J(\theta(y_1), \theta(y_2)) \\
= &~B^{v(T)}\cdot 2n^{v(Q)-2}\sum_{uv\in E(G)}d_J(u,v) = 2B^{v(T)}n^{v(Q)-2}\mathbb{E}[e(R_J)],
\end{align*}
where the last equality follows from Claim~\ref{cl:A2}.
Hence, for the fixed sequence $(I_1,\ldots,I_s)$, its contribution to $\hom(T,J^{+})$ is at most $\left(\prod_{i=1}^s x_i^{|I_i|}\right)2B^{v(T)}n^{v(Q)-2}\mathbb{E}[e(R_J)].$
Summing over all sequences $(I_1,\ldots,I_s)$ for which $T_s\cong Q$ and using the definition of $f_Q(\mathbf{x})$, we obtain that the contribution corresponding to $Q$ is at most
$2B^{v(T)}f_Q(\mathbf{x})n^{v(Q)-2}\mathbb{E}[e(R_J)].$
Finally, summing over all $Q\in \mathcal{Q}_T$ with $p(Q)=2$ and using the definition of $S(\mathbf{x},n)$, we find that the contribution to $\hom(T,J^{+})$ from the $A_2(\mathbf{x},G)$-part is at most
\begin{equation}\label{eq:A2}
B^{v(T)}S(\mathbf{x},n)\cdot \mathbb{E}[e(R_J)].
\end{equation}

Combining \eqref{eq:A0}, \eqref{eq:A3} and \eqref{eq:A2}, we obtain
$$\hom(T,J^+)\leq B^{v(T)}\Big(A_0(\mathbf{x},G)+S(\mathbf{x},n)\cdot \mathbb{E}[e(R_J)]+S(\mathbf{x},n)\cdot O_T\left(n^{-3}\right)\Big).$$
Since $\hom(T,J)\leq \hom(T,J^{+})$, the desired inequality follows.
\hfill$\blacksquare$

\vspace{0.2cm}

Let $k\colonequals \min_{F\in \mathcal{F}}\chi(F)\geq 3$.
In the remainder of this section, we assume that $G\in \mathcal{G}_k$ is an $n$-vertex graph.
Set $z\colonequals \cut_{k-1}(G)$.
Then Lemma~\ref{le:AlSS-NPhard} implies that
\begin{equation}\label{eq:z}
z=\cut_{k-1}(G)=\ex(G, K_2, K_k)
\end{equation}
and computing $z$ exactly for inputs $G\in\mathcal{G}_k$ is NP-hard.
We next establish the upper and lower bounds needed for the reduction.

\begin{lemma}\label{le:upper}
There exists a constant $\eta=\eta(T,\mathcal{F})>0$ such that every $\mathcal{F}$-free subgraph $J\subseteq Z$ satisfies
$$\hom(T, J)\leq B^{v(T)}\left(A_0(\mathbf{x},G) +S(\mathbf{x},n)\cdot z +S(\mathbf{x},n)\cdot O\left(n^{k-\eta}B^{-\eta} +n^{-3}\right)\right).$$
\end{lemma}

\noindent {\bf Proof.}
By Lemma~\ref{le:leaf-hom-4}, it suffices to prove that $\mathbb{E}[e(R_J)]\leq z+O(n^{k-\eta}B^{-\eta})$ whenever $J$ is $\mathcal{F}$-free.
Since deleting one edge from each copy of $K_k$ in $R_J$ yields a $K_k$-free subgraph of $G$, we have $\ex(G, K_2, K_k)\geq e(R_J)-\mathcal{N}(R_J, K_k)$.
Using~\eqref{eq:z}, we obtain
\begin{equation}\label{eq:upper-1}
e(R_J)\leq \ex(G, K_2, K_k)+\mathcal{N}(R_J, K_k)=z+\mathcal{N}(R_J, K_k).
\end{equation}

Let $F_0\in \mathcal{F}$ with $\chi(F_0)=k$.
Let $J^{-}$ be the subgraph of $J$ induced by $V(J)\cap \bigcup_{v\in V(G)}C_v$, that is, we only preserve the vertices and edges within the union of the clusters corresponding to $V(G)$.
Since $J$ is $\mathcal{F}$-free, $J^{-}\subseteq J$ and $F_0\in \mathcal{F}$, we have that $J^{-}$ is $F_0$-free.
Moreover, since $\chi(F_0)=k$, $F_0$ is a subgraph of a blow-up of $K_k$.
Then by Lemma~\ref{le:AlSh16-degen}, there exist constants $c>0$ and $0<\eta<k$ such that
$\mathcal{N}(J^{-},K_k)\leq cv(J^{-})^{k-\eta}\leq c(nB)^{k-\eta}.$
Every copy of $K_k$ in $J^{-}$ meets $k$ distinct clusters, so all its vertices are selected with probability $B^{-k}$.
After relabelling, the selected copies correspond bijectively to the copies of $K_k$ in $R_J$.
Hence,
\begin{equation*}\label{eq:upper-3}
\mathbb{E}[\mathcal{N}(R_J, K_k)]=B^{-k}\mathcal{N}(J^{-},K_k)\leq cB^{-k}(nB)^{k-\eta}=cn^{k-\eta}B^{-\eta}.
\end{equation*}
Taking expectations in \eqref{eq:upper-1}, we have $\mathbb{E}[e(R_J)]\leq z+\mathbb{E}[\mathcal{N}(R_J, K_k)]\leq z+cn^{k-\eta}B^{-\eta}.$
This completes the proof of Lemma~\ref{le:upper}.
\hfill$\blacksquare$


\begin{lemma}\label{le:lower}
For every $n$-vertex graph $H$ and every positive integer $B$, we have
$$\hom(T, L_{\mathbf{x}}(H)(B))\geq B^{v(T)}\Big(A_0(\mathbf{x}, H)+S(\mathbf{x},n)\cdot e(H)\Big).$$
\end{lemma}

\noindent {\bf Proof.}
Every homomorphism from $T$ to $L_{\mathbf{x}}(H)$ can be lifted to exactly $B^{v(T)}$ homomorphisms from $T$ to $L_{\mathbf{x}}(H)(B)$.
Thus the result follows from \eqref{eq:hom-3part}, Lemma~\ref{le:leaf-hom-3} and $A_{\geq 3}(\mathbf{x},H)\geq 0$.
\hfill$\blacksquare$


\subsection{Finishing the proof}
\label{subsec:proof_bipar}

In this subsection, we complete the proof of Theorem~\ref{thm:bipar}.
In order to apply Corollary~\ref{cor:transfer}, we first deduce a lower bound on $\ex'(Z, T, \mathcal{F})$ and an upper bound on $\ex(Z, T, \mathcal{F})$.
Let $\gamma$ be the number of automorphisms of $T$,
so each copy of $T$ corresponds to exactly $\gamma$ injective homomorphisms.
Then for any $N$-vertex graph $A$, we have
\begin{equation}\label{eq:aut}
0\leq \hom(T, A)-\gamma\mathcal{N}(A, T)\leq v(T)^{2}N^{v(T)-1}.
\end{equation}
To see this, note that every noninjective homomorphism maps two distinct vertices of $T$ to a single vertex of $A$, and there are $N$ choices for their common image.
Once these two vertices have been fixed, there are at most $N^{v(T)-2}$ choices for the images of the remaining vertices.
Thus $\hom(T, A)-\gamma\mathcal{N}(A, T)\leq {v(T)\choose 2}NN^{v(T)-2}\leq v(T)^{2}N^{v(T)-1}.$

Let $H$ be a spanning subgraph of $G$ with $\chi(H)\leq k-1$ and $e(H)=z$.
Using the same added vertices and clusters as in the construction of $Z$, we regard $L_{\mathbf{x}}(H)(B)$ as a spanning subgraph of $Z$.
Since $\chi(H)\leq k-1$, we have $\chi(L_{\mathbf{x}}(H))\leq k-1$ and $\chi(L_{\mathbf{x}}(H)(B))\leq k-1$.
Indeed, we can extend a fixed $(k-1)$-coloring of $H$ to $L_{\mathbf{x}}(H)$ and $L_{\mathbf{x}}(H)(B)$ by assigning each new leaf a color different from that of its anchor, and assigning every vertex of a blow-up cluster the color of its original vertex.
Then $L_{\mathbf{x}}(H)(B)$ is $\mathcal{F}$-free since $\min_{F\in \mathcal{F}}\chi(F)=k$.
Moreover, applying Lemma~\ref{le:lower} to $H$, we obtain
$$\hom(T, L_{\mathbf{x}}(H)(B))\geq B^{v(T)}\Big(A_0(\mathbf{x},H)+S(\mathbf{x},n)\cdot e(H)\Big).$$
By Lemma~\ref{le:leaf-hom-3}, we have $A_0(\mathbf{x},H)=A_0(\mathbf{x},G)$ since $v(H)=v(G)=n$.
Combining this equality and the preceding lower bound with \eqref{eq:aut}, we have
\begin{align}\label{eq:lower}
\ex'(Z, T, \mathcal{F}) \geq &~\mathcal{N}(L_{\mathbf{x}}(H)(B), T) \geq \gamma^{-1}\left(\hom(T, L_{\mathbf{x}}(H)(B))-v(T)^{2}v(L_{\mathbf{x}}(H)(B))^{v(T)-1}\right) \nonumber\\
\geq &~\gamma^{-1}\left(B^{v(T)}\Big(A_0(\mathbf{x},H)+S(\mathbf{x},n)\cdot e(H)\Big)-v(T)^2v(Z)^{v(T)-1}\right) \nonumber\\
\geq &~\gamma^{-1}B^{v(T)}\big(A_0(\mathbf{x},G)+S(\mathbf{x},n)\cdot z\big)-\gamma^{-1}v(T)^2v(Z)^{v(T)-1}.
\end{align}

On the other hand, let $J\subseteq Z$ be an $\mathcal{F}$-free subgraph with $\mathcal{N}(J,T)=\ex(Z, T, \mathcal{F})$.
Then by Lemma~\ref{le:upper} and Inequality~\eqref{eq:aut}, there exist constants $\eta=\eta(T,\mathcal{F})>0$ and $C=C(T,\mathcal{F})>0$ such that
\begin{align}\label{eq:upper}
\ex(Z, T, \mathcal{F}) = &~\mathcal{N}(J,T) \leq \gamma^{-1} \hom(T, J) \nonumber\\
\leq &~\gamma^{-1}B^{v(T)}\left(A_0(\mathbf{x},G) +S(\mathbf{x},n)\cdot z +S(\mathbf{x},n)C\left(n^{k-\eta}B^{-\eta} +n^{-3}\right)\right).
\end{align}

By Inequality~\eqref{eq:vZ}, we have $v(Z)\leq 2^sBn^{1+\sum\nolimits_{i=1}^{s}d_i}.$
Set $K\colonequals 1+\sum\nolimits_{i=1}^{s}d_i$.
Then $v(Z)\leq 2^sBn^{K}.$
Choose a fixed positive integer $\beta$ depending only on $T,\mathcal{F},\varepsilon$ so that
\begin{equation}\label{eq:c}
\beta>K(v(T)-1)+1, \qquad \beta\eta > k-\eta+1, \qquad \beta\varepsilon > K(v(T)-\varepsilon)+1.
\end{equation}
Let $B\colonequals n^{\beta}$.
Define
$$L\colonequals \gamma^{-1}B^{v(T)}A_0(\mathbf{x},G), \qquad S\colonequals \gamma^{-1}B^{v(T)}S(\mathbf{x},n),$$
$$a\colonequals \gamma^{-1}v(T)^2v(Z)^{v(T)-1}, \qquad b\colonequals \gamma^{-1}B^{v(T)}S(\mathbf{x},n) C\left(n^{k-\eta}B^{-\eta} +n^{-3}\right).$$
Then Inequalities~\eqref{eq:lower} and \eqref{eq:upper} yield
$$L+Sz-a \leq \ex'(Z, T, \mathcal{F})\leq \ex(Z, T, \mathcal{F})\leq L+Sz+b.$$
We now show that $a+b+v(Z)^{v(T)-\varepsilon}= o(S).$
Recall that $S(\mathbf{x},n)$ is an integer greater than $1$.
Combining this with $B= n^{\beta}$, $v(Z)\leq 2^sBn^{K}= 2^sn^{\beta+K}$ and \eqref{eq:c}, we have
\begin{align*}
S^{-1}a= &~\gamma B^{-v(T)}S(\mathbf{x},n)^{-1}\gamma^{-1}v(T)^2v(Z)^{v(T)-1} < v(T)^2 B^{-v(T)}(2^sn^{\beta+K})^{v(T)-1} \\
= &~O\left(n^{-\beta v(T)}\cdot n^{(\beta+K)(v(T)-1)}\right) = O\left(n^{K(v(T)-1)-\beta}\right)= O(n^{-1}) = o(1),
\end{align*}
\begin{align*}
S^{-1}b= &~\gamma B^{-v(T)}S(\mathbf{x},n)^{-1}\gamma^{-1}B^{v(T)}S(\mathbf{x},n) C\left(n^{k-\eta}B^{-\eta} +n^{-3}\right) \\
= &~C\left(n^{k-\eta-\beta\eta} + n^{-3}\right)= O(n^{-1}) = o(1),
\end{align*}
\begin{align*}
S^{-1}v(Z)^{v(T)-\varepsilon}\leq &~\gamma B^{-v(T)}S(\mathbf{x},n)^{-1}(2^sn^{\beta +K})^{v(T)-\varepsilon}< \gamma n^{-\beta v(T)}(2^sn^{\beta +K})^{v(T)-\varepsilon} \\
= &~O\left(n^{-\beta\varepsilon + K(v(T)-\varepsilon)}\right) = O(n^{-1}) = o(1).
\end{align*}
Thus
$$a+b+v(Z)^{v(T)-\varepsilon}=o(S).$$

Choose a fixed integer $N^{\ast}\geq n_0(k)$ so that $a+b+v(Z)^{v(T)-\varepsilon}<S/2$ whenever $n\geq N^{\ast}$,
and define $\mathcal{G}_{k,N^{\ast}}\colonequals \{G\in\mathcal G_k\colon\, v(G)\geq N^{\ast}\}.$
Computing $\cut_{k-1}(G)$ remains NP-hard on this class, since a polynomial-time algorithm on this class would extend to all of $\mathcal{G}_k$ by exhaustive search on the instances with fewer than $N^{\ast}$ vertices.
The integers $s, d_1, \ldots, d_s, \beta$ are fixed once $T, \mathcal{F},\varepsilon$ are fixed.
Thus the graph $Z$ can be constructed in polynomial time, and $v(Z)=O(n^{K+\beta})$.
The polynomials $f_Q$ depend only on $T$, and the formulas for $A_0(\mathbf{x},G)$ and $S(\mathbf{x},n)$ therefore show that $L$ and $S$ are polynomial-time computable rational numbers, with $L\geq 0$ and $S>0$.
Now by Corollary~\ref{cor:transfer}, each of the two approximation problems is NP-hard even when the input graph is restricted to $\{Z(G)\colon\, G\in\mathcal{G}_{k,N^{\ast}}\}$, and hence also for arbitrary host graphs.
This completes the proof of Theorem~\ref{thm:bipar}.

\section{Proofs of Theorems~\ref{thm:multipar} and \ref{thm:chromatic+2}}
\label{sec:pf_nonbipar}

In this section, we prove Theorem~\ref{thm:multipar} and establish two further general hardness results.
For an edge $e=uv$ of $T$, let $T/e$ be the graph obtained from $T$ by contracting the edge $e$, that is, deleting the edge $e$ and replacing vertices $u$ and $v$ by a single vertex adjacent precisely to the vertices which were adjacent to either $u$ or $v$ in $T$.
We first prove the following more general result. Theorem~\ref{thm:chromatic+2} is an immediate corollary, since $\chi(T/e)\leq \chi(T)+1$.

\begin{theorem}\label{thm:chromatic}
Let $T$ be a graph with $e(T)>0$, $\mathcal{F}$ be a nonempty family of graphs such that $\min_{F\in \mathcal{F}}\chi(F)=k\geq 3$, and $\varepsilon>0$.
If $T$ contains an edge $e$ satisfying $\chi(T)\leq k-1$ and $\chi(T/e)\leq k-1$, then for each of the two parameters $\ex(G, T, \mathcal{F})$ and $\ex'(G, T, \mathcal{F})$, approximation within an additive error of $n^{v(T)-\varepsilon}$ is NP-hard for an input graph $G$ on $n$ vertices.
\end{theorem}

We prove Theorems~\ref{thm:multipar} and \ref{thm:chromatic} via Corollary~\ref{cor:transfer}.
The constructions differ from the one used in the proof of Theorem~\ref{thm:bipar}.
The symmetric rooted construction also yields the following general hardness result for forbidden families whose members each contain $K_k$.
In particular, Conjecture~\ref{conj:AlSh22} for $\ex(G, T, K_k)$ follows by taking $\mathcal{F}=\{K_k\}$.

\begin{theorem}\label{thm:clique}
Let $T$ be a graph with $e(T)>0$, $\mathcal{F}$ be a nonempty family of graphs such that $\min_{F\in \mathcal{F}}\chi(F)=k> \omega(T)$ and every $F\in \mathcal{F}$ contains $K_k$ as a subgraph, and $\varepsilon>0$.
Then for each of the two parameters $\ex(G, T, \mathcal{F})$ and $\ex'(G, T, \mathcal{F})$, approximation within an additive error of $n^{v(T)-\varepsilon}$ is NP-hard for an input graph $G$ on $n$ vertices.
\end{theorem}

In the remainder of this section, we prove Theorems~\ref{thm:multipar}, \ref{thm:chromatic} and \ref{thm:clique}.
In Section~\ref{subsec:root}, we first introduce the symmetric rooted construction and establish the estimates needed for Theorems~\ref{thm:chromatic} and \ref{thm:clique}, which we prove in Section~\ref{subsec:pf_cliquer+chromatic}.
In Section~\ref{subsec:pf_multipar}, we adapt the construction to complete multipartite graphs and prove Theorem~\ref{thm:multipar}.

\subsection{Symmetric rooted extensions}
\label{subsec:root}

For Theorems~\ref{thm:chromatic} and \ref{thm:clique}, the case $v(T)=2$ (i.e., $T=K_2$) follows from the work of Alon, Shapira and Sudakov~\cite{AlSS}.
Thus we may assume that $v(T)\geq 3$ throughout this section.
Fix a graph $T$ with $v(T)\geq 3$, and fix an oriented edge $xy\in E(T)$.
Let $G$ be an $n$-vertex graph with $e(G)>0$, and let $A$ and $B$ be positive integers.
For convenience, we restate the construction of $Z_{A,B}(G, T, xy)$ given in Section~\ref{subsec:sketch}.
We construct the graph $Z_{A,B}(G, T, xy)$ from $G$ in the following steps (see Figure~\ref{fig:root} for an illustration):
\begin{enumerate}
\item[(1)] We begin with the empty graph on $V(G)$. For every ordered pair $(u,v)$ of distinct vertices of $G$, take a copy $T_{u,v}$ of the graph $T-\{x,y\}$, where all these copies are pairwise disjoint and are disjoint from $V(G)$.
\item[(2)] For each ordered pair $(u,v)$, replace $T_{u,v}$ by its $A$-blow-up $T_{u,v}(A)$.
\item[(3)] For each ordered pair $(u,v)$, regard $u$ and $v$ as playing the roles of $x$ and $y$, respectively, and add edges as follows: for each vertex $w\in V(T-\{x,y\})$,
    \begin{itemize}
        \item if $xw\in E(T)$, add all edges between $u$ and the cluster corresponding to $w$ in $T_{u,v}(A)$;
        \item if $yw\in E(T)$, add all edges between $v$ and the cluster corresponding to $w$ in $T_{u,v}(A)$.
    \end{itemize}
    No edge between $u$ and $v$ is added at this step.
\item[(4)] Add all edges of $G$ among the original vertices of $G$.
\item[(5)] Take the $B$-blow-up of the resulting graph.
\end{enumerate}
We simply denote $Z_{A,B}(G, T, xy)$ by $Z$ when the parameters are clear.
Recall that for each $u\in V(G)$, $C_u$ is the cluster corresponding to $u$, so $|C_u|=B$.
We call the vertices arising from $V(G)$ the {\it source vertices}, and the other vertices of $Z$ the {\it auxiliary vertices}.
Hence, $Z$ has exactly $nB$ source vertices and $n(n-1)(v(T)-2)AB$ auxiliary vertices, so
\begin{equation}\label{eq:symmetric-rooted-vZ}
v(Z)=nB+n(n-1)(v(T)-2)AB \leq v(T)n^2AB.
\end{equation}
Moreover, the subgraph of $Z$ induced by the source vertices is precisely the $B$-blow-up $G(B)$ of $G$.
Let $Z_0$ be the spanning subgraph of $Z$ obtained by deleting all edges of $G(B)$, and let $$L\colonequals \mathcal{N}(Z_0,T).$$
For an edge $e\in E(G(B))$, let $X$ be the number of copies of $T$ in $Z_0+e$ that contain $e$ and have exactly two source vertices, namely the endpoints of $e$.
By symmetry, $X$ is independent of the choice of $e$.
To obtain copies counted by $X$, fix an orientation of $e$ and use its endpoints for $x$ and $y$, respectively.
Then choose one vertex from each of the $v(T)-2$ auxiliary sets rooted at the corresponding ordered pair of vertices of $G$.
Each of these sets has size $AB$, and distinct choices give distinct copies.
Consequently,
\begin{equation}\label{eq:symmetric-rooted-X}
X\geq (AB)^{v(T)-2}>0.
\end{equation}

\begin{lemma}\label{lem:symmetric-rooted-Z0+H}
For every spanning subgraph $H\subseteq G(B)$, we have $$L+Xe(H)\leq \mathcal{N}(Z_0+H,T)=L+Xe(H)+O_T\left(n^{2v(T)-3}A^{v(T)-3}B^{v(T)}\right).$$
\end{lemma}

\noindent {\bf Proof.}
The copies of $T$ in $Z_0+H$ using no edge of $H$ are precisely those counted by $L$.
A copy using an edge of $H$ and containing exactly two source vertices contains exactly one edge of $H$.
For each edge of $H$, there are exactly $X$ such copies, giving the contribution $Xe(H)$.

For $j\in \{3, \ldots, v(T)\}$, the number of copies of $T$ containing exactly $j$ source vertices is at most
$O_T\left((nB)^{j}(n^2AB)^{v(T)-j}\right)=O_T\left(n^{2v(T)-j}A^{v(T)-j}B^{v(T)}\right)$.
Thus the number of copies of $T$ containing at least three source vertices is at most
$O_T\left(\sum_{j=3}^{v(T)}n^{2v(T)-j}A^{v(T)-j}B^{v(T)}\right)=O_T\left(n^{2v(T)-3}A^{v(T)-3}B^{v(T)}\right).$
The result follows.
\hfill$\blacksquare$

\vspace{0.2cm}

Let $k\geq 3$.
Throughout this section, we assume that $G$ is an $n$-vertex graph in $\mathcal{G}_k$.
Set $z\colonequals \cut_{k-1}(G)$.
Then Lemma~\ref{le:AlSS-NPhard} implies that
\begin{equation}\label{eq:symmetric-rooted-z}
z=\cut_{k-1}(G)=\ex(G, K_2, K_k),
\end{equation}
and computing $\cut_{k-1}(G)$ is NP-hard for inputs $G\in \mathcal{G}_k$.
We next obtain an upper bound by adapting the arguments of Lemmas~\ref{le:leaf-hom-4} and \ref{le:upper}.

\begin{lemma}\label{lem:symmetric-rooted-NJT}
Let $\mathcal{F}$ be a fixed nonempty family of graphs with $\min_{F\in \mathcal{F}}\chi(F)=k\geq 3$.
There exists a constant $\eta=\eta(T,\mathcal{F})>0$ such that every $\mathcal{F}$-free subgraph $J\subseteq Z$ satisfies
$$\mathcal{N}(J,T) \leq L+XB^2z +O_{T, \mathcal{F}}\left(Xn^{k-\eta}B^{2-\eta} + n^{2v(T)-3}A^{v(T)-3}B^{v(T)}\right).$$
\end{lemma}

\noindent {\bf Proof.}
Let $J^{-}$ be the subgraph of $J$ induced by $V(J)\cap V(G(B))$, so $E(J^{-})= E(J)\cap E(G(B))$.
Let $H$ be the spanning subgraph of $G(B)$ with $E(H)=E(J^-)$, and let $J^{+}\colonequals Z_0+H$.
Then $J\subseteq J^{+}\subseteq Z$, and by Lemma~\ref{lem:symmetric-rooted-Z0+H}, we have
$$\mathcal{N}(J,T) \leq ~\mathcal{N}(J^{+},T) \leq L+Xe(J^{-})+O_T\left(n^{2v(T)-3}A^{v(T)-3}B^{v(T)}\right).$$
Now it suffices to show that $e(J^{-})\leq B^2z+cn^{k-\eta}B^{2-\eta}$ for some constants $c>0$ and $\eta>0$ depending only on $T$ and $\mathcal{F}$.

Let $F_0\in \mathcal{F}$ with $\chi(F_0)=k$.
Since $J$ is $\mathcal{F}$-free and $J^{-}\subseteq J$, we have that $J^{-}$ is $F_0$-free.
Moreover, $F_0$ is a subgraph of a blow-up of $K_k$ since $\chi(F_0)=k$.
Then by Lemma~\ref{le:AlSh16-degen}, there exist constants $c>0$ and $0<\eta<k$ such that
$\mathcal{N}(J^{-}, K_k)\leq c(nB)^{k-\eta}.$

Define a random spanning subgraph $R_J$ of $G$ as follows.
We choose one vertex $v'\in C_v$ independently and uniformly for each $v\in V(G)$.
For each $v \in V(G)$, if $v'\notin V(J)$, then $v$ is isolated in $R_J$.
For every $uv\in E(G)$, we put $uv\in E(R_J)$ if and only if $u', v'\in V(J)$ and $u'v'\in E(J)$.
Note that for every edge in $J^{-}$, the probability that its corresponding edge is present in $R_J$ is $B^{-2}$,
so $\mathbb{E}[e(R_J)]=B^{-2}e(J^{-})$.
Moreover, each copy of $K_k$ in $J^-$ meets $k$ distinct source clusters, since every cluster is independent.
Its vertices are all selected with probability $B^{-k}$.
The map sending the selected vertex in $C_v$ to $v$ gives a bijection between the copies of $K_k$ in $J^-$ whose vertices are all selected and the copies of $K_k$ in $R_J$.
Hence, $\mathbb{E}[\mathcal{N}(R_J,K_k)] = B^{-k}\mathcal{N}(J^{-},K_k)\leq cB^{-k}(nB)^{k-\eta}.$

Since deleting one edge from each copy of $K_k$ in $R_J$ yields a $K_k$-free subgraph of $G$, we have $\ex(G, K_2, K_k)\geq e(R_J)-\mathcal{N}(R_J, K_k)$.
Combining this inequality with \eqref{eq:symmetric-rooted-z}, we obtain
$e(R_J)\leq \ex(G, K_2, K_k)+\mathcal{N}(R_J, K_k)=z+\mathcal{N}(R_J, K_k).$
Then $$e(J^{-})=B^{2}\mathbb{E}[e(R_J)]\leq B^{2}z+ B^{2}\mathbb{E}[\mathcal{N}(R_J, K_k)] \leq B^{2}z+ B^{2}cB^{-k}(nB)^{k-\eta} = B^2z+cn^{k-\eta}B^{2-\eta}.$$
This completes the proof of Lemma~\ref{lem:symmetric-rooted-NJT}.
\hfill$\blacksquare$

\vspace{0.2cm}

The following result will be used in our proofs of Theorems~\ref{thm:chromatic} and \ref{thm:clique}.

\begin{lemma}\label{lem:symmetric-rooted-transfer}
Fix a graph $T$ with $v(T)\geq 3$, a nonempty family $\mathcal{F}$ with $\min_{F\in\mathcal{F}}\chi(F)=k\geq 3$, and an oriented edge $xy\in E(T)$.
Suppose that for every $G\in\mathcal G_k$ and all positive integers $A,B$, there is a spanning subgraph $H\subseteq G$ with $\chi(H)\leq k-1$ and $e(H)=\cut_{k-1}(G)$ such that $Z_0+H(B)$ is $\mathcal{F}$-free.
Then for every $\varepsilon>0$, there exists a positive integer $\beta=\beta(T, \mathcal{F},\varepsilon)$ such that the following holds.
For each $G\in\mathcal G_k$, write $n=v(G)$ and set $Z(G)\colonequals Z_{n^{2v(T)-2},\,n^\beta}(G,T,xy)$.
It is NP-hard to approximate $\ex(Z,T,\mathcal{F})$ or $\ex'(Z,T,\mathcal{F})$ within an additive error of $v(Z)^{v(T)-\varepsilon}$ for input graphs $Z\in\{Z(G)\colon\, G\in \mathcal{G}_k\}$.
\end{lemma}

\noindent {\bf Proof.}
By Lemma~\ref{lem:symmetric-rooted-NJT}, we have
\begin{equation}\label{eq:clique-upper}
\ex(Z, T, \mathcal{F})\leq L+XB^2z +CXn^{k-\eta}B^{2-\eta} + Cn^{2v(T)-3}A^{v(T)-3}B^{v(T)}
\end{equation}
for some constants $\eta=\eta(T,\mathcal{F})>0$ and $C=C(T,\mathcal{F})>0$.
On the other hand, note that $Z_0+H(B)$ is a spanning subgraph of $Z$ and $H(B)$ is a spanning subgraph of $G(B)$.
Since $Z_0+H(B)$ is $\mathcal{F}$-free, we have
\begin{equation}\label{eq:clique-lower}
\ex'(Z, T, \mathcal{F})\geq \mathcal{N}(Z_0+H(B),T)\geq L+Xe(H(B)) = L+XB^2z
\end{equation}
by applying Lemma~\ref{lem:symmetric-rooted-Z0+H} to $H(B)$.

We now choose the parameters.
Let $A\colonequals n^{2v(T)-2}$ and let $B\colonequals n^{\beta}$, where $\beta$ is a fixed positive integer depending only on $T, \mathcal{F}, \varepsilon$, chosen so that
\begin{equation}\label{eq:clique-parameters}
    \beta\eta>k-\eta+1, \qquad \beta\varepsilon > 2(v(T)-\varepsilon)+(2v(T)-2)(2-\varepsilon)+1.
\end{equation}
Define
$$S\colonequals XB^2, \qquad a\colonequals 0, \qquad b\colonequals CXn^{k-\eta}B^{2-\eta} + Cn^{2v(T)-3}A^{v(T)-3}B^{v(T)}.$$
Then Inequalities~\eqref{eq:clique-upper} and \eqref{eq:clique-lower} yield
$$L+Sz-a \leq \ex'(Z, T, \mathcal{F})\leq \ex(Z, T, \mathcal{F})\leq L+Sz+b.$$
Using the definitions of $S$ and $b$, together with \eqref{eq:symmetric-rooted-vZ}, \eqref{eq:symmetric-rooted-X} and \eqref{eq:clique-parameters}, we have
\begin{align*}
S^{-1}b= &~X^{-1}B^{-2}\left(CXn^{k-\eta}B^{2-\eta} + Cn^{2v(T)-3}A^{v(T)-3}B^{v(T)}\right) \\
\leq &~Cn^{k-\eta}B^{-\eta} + C(AB)^{-(v(T)-2)}n^{2v(T)-3}A^{v(T)-3}B^{v(T)-2} \\
= &~Cn^{k-\eta}B^{-\eta} + Cn^{2v(T)-3}A^{-1} = Cn^{k-\eta-\beta\eta} + Cn^{2v(T)-3-(2v(T)-2)} = O(n^{-1}) = o(1),
\end{align*}
\begin{align*}
S^{-1}v(Z)^{v(T)-\varepsilon}\leq &~ X^{-1}B^{-2}(v(T)n^2AB)^{v(T)-\varepsilon}\leq (AB)^{-(v(T)-2)}B^{-2}(v(T)n^2AB)^{v(T)-\varepsilon} \\
= &~O\left(n^{2(v(T)-\varepsilon)}A^{2-\varepsilon}B^{-\varepsilon}\right) =O\left(n^{2(v(T)-\varepsilon)+(2v(T)-2)(2-\varepsilon)-\beta\varepsilon}\right) = O(n^{-1}) = o(1).
\end{align*}
Thus
$$a+b+v(Z)^{v(T)-\varepsilon}=o(S).$$

Choose a fixed integer $N^{\ast}\geq n_0(k)$ so that $a+b+v(Z)^{v(T)-\varepsilon}<S/2$ whenever $n\geq N^{\ast}$, and define $\mathcal{G}_{k,N^{\ast}}\colonequals \{G\in\mathcal G_k\colon\, v(G)\geq N^{\ast}\}.$
Computing $\cut_{k-1}(G)$ remains NP-hard on this class, since an exact polynomial-time algorithm on this class would extend to all of $\mathcal G_k$ by exhaustive search on the instances with fewer than $N^{\ast}$ vertices.
Since $v(Z)\leq v(T)n^2AB$$, A=n^{2v(T)-2}$ and $B=n^{\beta}$, where $\beta$ is fixed once $T, \mathcal{F},\varepsilon$ are fixed, we have $v(Z)=O(n^{2v(T)+\beta})$.
Hence, $Z$ can be constructed in polynomial time.
Since $T$ is fixed, $L=\mathcal{N}(Z_0,T)$ and $X$ can be computed in polynomial time by enumerating copies of $T$ in $Z_0$ and $Z_0+e$ and checking the conditions defining $X$.
Thus $L\geq 0$ and $S=XB^2>0$ are polynomial-time computable integers.
Now by Corollary~\ref{cor:transfer}, each of the two approximation problems is NP-hard even when the input graph is restricted to $\{Z(G)\colon\, G\in \mathcal{G}_{k,N^{\ast}}\}$.
Consequently, the same holds on $\{Z(G)\colon\, G\in \mathcal{G}_k\}$ and on arbitrary host graphs.
This completes the proof of Lemma~\ref{lem:symmetric-rooted-transfer}.
\hfill$\blacksquare$


\subsection{Proofs of Theorems~\ref{thm:chromatic} and \ref{thm:clique}}
\label{subsec:pf_cliquer+chromatic}

In this subsection, we prove Theorems~\ref{thm:chromatic} and \ref{thm:clique} using Lemma~\ref{lem:symmetric-rooted-transfer}.
\vspace{0.2cm}

\noindent {\bf Proof of Theorem~\ref{thm:chromatic}.}
By assumption, we have $\min_{F\in \mathcal{F}}\chi(F)=k\geq 3$, and $T$ contains an edge $e=xy$ satisfying $\chi(T)\leq k-1$ and $\chi(T/e)\leq k-1$.
Orient $e$ as $xy$.
Let $G\in\mathcal G_k$, let $A,B$ be arbitrary positive integers, and set $Z=Z_{A,B}(G,T,xy)$.
Let $H$ be a spanning subgraph of $G$ with $\chi(H)\leq k-1$ and $e(H)=z$.
By Lemma~\ref{lem:symmetric-rooted-transfer}, it suffices to show that $Z_0+H(B)$ is $\mathcal{F}$-free.
To prove this, we first show that $\chi(Z_0+H(B))\leq k-1$.
Fix a proper $(k-1)$-coloring $\chi_{H}$ of $H$.
For each vertex $v\in V(H)=V(G)$, we assign every source vertex in the cluster $C_v$ the color $\chi_H(v)$ of $v$.
It remains to extend this coloring to each auxiliary vertex.
Consider the $(v(T)-2)AB$ auxiliary vertices rooted at an ordered pair $(u,v)$.
If $u$ and $v$ receive distinct colors in $\chi_H$, then take a proper $(k-1)$-coloring $\chi_T$ of $T$.
Since $x$ and $y$ are adjacent in $T$, they receive distinct colors in $\chi_T$.
By renaming the colors if necessary, we may assume that the colors of $u$ and $v$ are precisely the colors of $x$ and $y$, respectively.
Then for each vertex $w\in V(T-\{x,y\})$, we assign color $\chi_T(w)$ to every auxiliary vertex corresponding to $w$.
If $u$ and $v$ receive the same color, then take a proper $(k-1)$-coloring $\chi_{T/e}$ of $T/e$.
By renaming the colors if necessary, we may assume that the vertex obtained by contracting $e=xy$ receives the same color as $u$ and $v$.
Then for each vertex $w\in V(T-\{x,y\})$, we assign color $\chi_{T/e}(w)$ to every auxiliary vertex corresponding to $w$.
By the definition of $Z_{A,B}(G, T, xy)$, the sets of auxiliary vertices rooted at different ordered pairs are disjoint, and there are no edges between these sets.
Therefore, this gives a proper $(k-1)$-coloring of $Z_0+H(B)$.
Since $\min_{F\in \mathcal{F}}\chi(F)=k$, the graph $Z_0+H(B)$ is $\mathcal{F}$-free.
This completes the proof.
\hfill$\blacksquare$

\vspace{0.2cm}

\noindent {\bf Proof of Theorem~\ref{thm:clique}.}
Fix any oriented edge $xy\in E(T)$.
Let $G\in\mathcal G_k$, let $A,B$ be arbitrary positive integers, and set $Z=Z_{A,B}(G,T,xy)$.
Choose a spanning subgraph $H\subseteq G$ with $\chi(H)\leq k-1$ and $e(H)=z$.
By Lemma~\ref{lem:symmetric-rooted-transfer}, it suffices to prove that $Z_0+H(B)$ is $\mathcal{F}$-free.
To prove this, we first show that $Z_0+H(B)$ is $K_k$-free.
Recall that $Z_0$ contains no edge of $G(B)$.
Moreover, $H(B)$ contains no $K_k$ since $\chi(H)\leq k-1$.
Hence, if $Z_0+H(B)$ contains a $K_k$, then such a $K_k$ must contain an auxiliary vertex, say one rooted at $(u,v)$.
By the definition of $Z_{A,B}(G, T, xy)$, the sets of auxiliary vertices rooted at different ordered pairs are disjoint, and there are no edges between these sets.
Thus the copy of $K_k$ lies in $C_u\cup C_v$ together with the auxiliary vertices rooted at $(u,v)$.
The graph induced on these vertices is a subgraph of a blow-up of $T$, whose clique number is $\omega(T)< k$.
This is a contradiction.
Therefore, $Z_0+H(B)$ is $K_k$-free.
By the assumption, every member of $\mathcal{F}$ contains $K_k$, so the graph $Z_0+H(B)$ is $\mathcal{F}$-free.
This completes the proof.
\hfill$\blacksquare$

\subsection{Proof of Theorem~\ref{thm:multipar}}
\label{subsec:pf_multipar}

Let $T=K_{s_1, \ldots, s_r}$ with $r\geq 3$ and $s_r\geq \cdots \geq s_1\geq 1$.
Put $p\colonequals s_1+s_2$ and $q\colonequals s_3+\cdots +s_r$.
Since $K_r\subseteq T$ and $\chi(T)=r$, a graph $F$ is a subgraph of a blow-up of $T$ if and only if $\chi(F)\leq r$.
Thus the assumption that no $F\in \mathcal{F}$ is a subgraph of a blow-up of $T$ is equivalent to $\min_{F\in \mathcal{F}}\chi(F)>r$.
Let $k\colonequals \min_{F\in \mathcal{F}}\chi(F)$, so $k>r$.
Let $G\in \mathcal{G}_k$ be an $n$-vertex graph, and set $z\colonequals \cut_{k-1}(G)$.
Then Lemma~\ref{le:AlSS-NPhard} implies that
\begin{equation}\label{eq:multipar-z}
z=\cut_{k-1}(G)=\ex(G, K_2, K_k)
\end{equation}
and it is NP-hard to compute $z$ exactly.

We first construct the graph $Z$.
Let $A$ and $B$ be positive integers with $B\geq s_r$.
Start with $G$.
For every edge $uv\in E(G)$, introduce $r-2$ pairwise disjoint independent sets $V_{uv}^3, \ldots, V_{uv}^r,$ each of size $A$, where all these sets are disjoint from $V(G)$ and from all sets introduced for other edges.
Add all edges so that $\{u\}, \{v\}, V_{uv}^3, \ldots, V_{uv}^r$ form a blow-up of $K_r$.
No other edges incident with the new vertices are added.
Finally, let $Z$ be the $B$-blow-up of the resulting graph.

As before, the vertices arising from $V(G)$ are called {\it source vertices}.
For each $u\in V(G)$, let $C_u$ be the corresponding cluster, so $|C_u|=B$.
After the final blow-up, each set $V_{uv}^i$ gives rise to an independent set $C_{uv}^{i}$ of size $AB$.
Let $C_{uv}=\bigcup_{3\leq i\leq r}C_{uv}^{i}$.
We call the vertices in $C_{uv}$ auxiliary vertices corresponding to the edge $uv$.
Note that auxiliary vertices corresponding to distinct edges of $G$ are nonadjacent.
Since $e(G)\leq n^2/2$, we have
\begin{equation*}
    v(Z) = nB+e(G)(r-2)AB = O_r(n^2AB).
\end{equation*}

Let $\mathcal{P}$ be the set of distinct permutations $\mathbf{a}=(a_1,\ldots,a_r)$ of the multiset $\{s_1,\ldots,s_r\}$ satisfying $a_1+a_2=p$.
Define $$S\colonequals \sum_{\mathbf{a}\in \mathcal{P}}{B\choose a_1}{B\choose a_2}\prod_{i=3}^{r}{AB \choose a_i}.$$
For each $\mathbf{a} \in \mathcal{P}$, we have $\sum_{i=3}^{r} a_i=q$.
Since $\mathcal{P}$ is a fixed nonempty set and $B\geq s_r$, we have $S=\Theta_T\left(A^{q}B^{v(T)}\right).$
We next state and prove three claims concerning the number of copies of $T$ in $Z$.

\begin{claim}\label{cl:multipar-T-1}
Every copy of $T$ in $Z$ contains at most $q$ auxiliary vertices.
\end{claim}

\begin{proof}
Let $T'$ be a copy of $T$ in $Z$, and assume that $T'$ contains an auxiliary vertex.
First, assume that $T'$ contains auxiliary vertices corresponding to at least two edges of $G$.
Since auxiliary vertices corresponding to distinct edges of $G$ are nonadjacent, all auxiliary vertices of $T'$ must belong to the same partite set of $T'$.
Thus the number of such vertices is at most $s_r\leq q$.
Next, assume that all auxiliary vertices of $T'$ belong to one $C_{uv}$.
Since every $C_{uv}^{i}$ is an independent set, it meets at most one partite set of $T'$.
Consequently, the auxiliary vertices lie in at most $r-2$ partite sets of $T'$.
Therefore, the number of auxiliary vertices of $T'$ is at most $s_3+\cdots +s_r=q$.
\end{proof}

\begin{claim}\label{cl:multipar-T-2}
For each copy of $T$ containing exactly $q$ auxiliary vertices, all its auxiliary vertices belong to one $C_{uv}$.
Moreover, for each $uv\in E(G)$, the number of copies of $T$ containing exactly $q$ auxiliary vertices of $C_{uv}$ is $S$.
\end{claim}

\begin{proof}
Let $T'$ be a copy of $T$ containing exactly $q$ auxiliary vertices in $Z$.
We first show that all its auxiliary vertices belong to one $C_{uv}$.
Suppose for a contradiction that $T'$ contains auxiliary vertices corresponding to at least two edges of $G$.
Then, by the same argument as in the proof of Claim~\ref{cl:multipar-T-1}, all $q$ auxiliary vertices of $T'$ must belong to the same partite set of $T'$.
For $r\geq4$, the number of such vertices is at most $s_r= q-(s_3+\cdots+s_{r-1})< q$, a contradiction.
For $r=3$, we have $q=s_3$, so all $q$ auxiliary vertices of $T'$ must belong to the partite set of size $s_3$.
Let $u_1v_1$ and $u_2v_2$ be two edges of $G$ such that $T'$ contains auxiliary vertices from $C_{u_1v_1}$ and $C_{u_2v_2}$.
Every vertex in the other two partite sets is a source vertex adjacent to these auxiliary vertices, so it must belong to a cluster corresponding to a vertex of $\{u_1,v_1\}\cap\{u_2,v_2\}$.
Since the two edges $u_1v_1$ and $u_2v_2$ are distinct, there is at most one vertex in $\{u_1,v_1\}\cap\{u_2,v_2\}$.
However, every cluster is an independent set, so it cannot contain vertices from two partite sets of $T'$, a contradiction.
Hence, all auxiliary vertices of $T'$ belong to one $C_{uv}$.

Fix $uv\in E(G)$, and suppose that all $q$ auxiliary vertices of $T'$ lie in $C_{uv}$.
Since each independent set $C_{uv}^i$ meets at most one partite set of $T'$, these vertices lie in at most $r-2$ partite sets, whose total size is at most $s_3+ \cdots +s_r=q$.
Thus they constitute exactly $r-2$ entire partite sets of sizes $s_3, \ldots, s_r$, one contained in each of $C_{uv}^3, \ldots, C_{uv}^r$.
The remaining two partite sets consist of source vertices adjacent to all these auxiliary vertices, so they lie in $C_u\cup C_v$.
Since $C_u$ and $C_v$ are independent sets, one partite set lies entirely in $C_u$ and the other entirely in $C_v$.
Consequently, every such copy is obtained uniquely by choosing subsets of sizes $a_1, \ldots, a_r$ in $C_u, C_v, C_{uv}^3, \ldots, C_{uv}^r$, respectively, for some $(a_1, \ldots, a_r)\in\mathcal{P}$.
Conversely, every such choice gives a copy of $T$.
Therefore, the number of copies of $T$ containing exactly $q$ auxiliary vertices of $C_{uv}$ is
$$\sum_{\mathbf{a}\in \mathcal{P}}{|C_u|\choose a_1}{|C_v|\choose a_2}\prod_{i=3}^{r}{|C_{uv}^{i}| \choose a_i}=\sum_{\mathbf{a}\in \mathcal{P}}{B\choose a_1}{B\choose a_2}\prod_{i=3}^{r}{AB \choose a_i}=S.$$
This completes the proof.
\end{proof}

\begin{claim}\label{cl:multipar-T-3}
The number of copies of $T$ containing at most $q-1$ auxiliary vertices in $Z$ is at most $$O_T\left(n^{v(T)+q-1}A^{q-1}B^{v(T)}\right).$$
\end{claim}

\begin{proof}
The graph $Z$ has $nB$ source vertices and $O_r(n^2AB)$ auxiliary vertices.
For each $j\in \{0, 1, \ldots, q-1\}$, the number of copies of $T$ containing exactly $j$ auxiliary vertices is at most $O_T\left((nB)^{v(T)-j}(n^2AB)^{j}\right)=O_T\left(n^{v(T)+j}A^{j}B^{v(T)}\right)= O_T\left(n^{v(T)+q-1}A^{q-1}B^{v(T)}\right).$
Summing over $j\in \{0, 1, \ldots, q-1\}$, we obtain that the number of copies of $T$ containing at most $q-1$ auxiliary vertices is at most $O_T\left(n^{v(T)+q-1}A^{q-1}B^{v(T)}\right).$
\end{proof}

We next establish an upper bound on $\ex(Z, T, \mathcal{F})$ and a lower bound on $\ex'(Z, T, \mathcal{F})$.

\begin{claim}\label{cl:multipar-upper}
There exist constants $c=c(T,\mathcal{F})>0$, $C=C(T,\mathcal{F})>0$ and $\eta=\eta(T,\mathcal{F})>0$ such that
$$\ex(Z, T, \mathcal{F})\leq Sz +S \frac{cv(Z)^{w-\eta}}{A^{w-k}B^{w}} +Cn^{v(T)+q-1}A^{q-1}B^{v(T)},$$ where $w=k+(r-2){k\choose 2}$.
\end{claim}

\begin{proof}
Let $J$ be an $\mathcal{F}$-free subgraph of $Z$ with $\mathcal{N}(J,T)=\ex(Z, T, \mathcal{F})$.
Let $Y(J)$ denote the set of copies of $T$ in $J$ containing exactly $q$ auxiliary vertices.
By Claim~\ref{cl:multipar-T-3}, it suffices to show that for some constants $c=c(T,\mathcal{F})>0$ and $\eta=\eta(T,\mathcal{F})>0$, we have
\begin{equation}\label{eq:multipar-YJ}
    |Y(J)| \leq Sz +S \frac{cv(Z)^{w-\eta}}{A^{w-k}B^{w}},
\end{equation}
where $w=k+(r-2){k\choose 2}$.

Define a random spanning subgraph $R_J$ of $G$ on the vertex set $V(G)$ as follows.
For every $u\in V(G)$, choose a vertex $u'\in C_u$ uniformly at random.
For every $uv\in E(G)$ and every $i\in\{3,\ldots,r\}$, choose a vertex $w_{uv}^i\in C_{uv}^i$ uniformly at random, with all choices mutually independent.
Put $uv\in E(R_J)$ if and only if $u', v', w_{uv}^3, \ldots, w_{uv}^r$ are vertices of $J$ and induce a copy of $K_r$ in $J$.

For $uv\in E(G)$ and $\mathbf a=(a_1,\ldots,a_r)\in\mathcal P$, let $Y_{uv,\mathbf a}(J)$ consist of the copies in $Y(J)$ having $a_1$ vertices in $C_u$, $a_2$ vertices in $C_v$, and $a_i$ vertices in $C_{uv}^i$ for each $3\leq i\leq r$.
By Claim~\ref{cl:multipar-T-2}, we have $|Y(J)|=\sum_{uv\in E(G)}\sum_{\mathbf{a}\in \mathcal{P}}|Y_{uv,\mathbf{a}}(J)|$.
We now prove an upper bound on $|Y_{uv,\mathbf{a}}(J)|$.
Choose uniformly and independently an $a_1$-subset of $C_u$, an $a_2$-subset of $C_v$, and an $a_i$-subset of $C_{uv}^{i}$ for every $i\in\{3,\ldots,r\}$.
Then from each chosen subset, choose one representative vertex uniformly at random.
For any one of the sets $C_u,C_v,C_{uv}^3,\ldots,C_{uv}^r$, if its size is $M$ and the chosen subset has size $a$, then each vertex is the representative with probability $\frac{a}{M}\cdot\frac{1}{a}=\frac{1}{M}$.
Thus the representatives are independent and uniformly distributed over $C_u, C_v, C_{uv}^3, \ldots, C_{uv}^r$, respectively, exactly as in the definition of $R_J$.
Moreover, if the chosen subsets form a copy of $T$ in $J$, then the $r$ representative vertices form a copy of $K_r$ in $J$.
Therefore, we have
$|Y_{uv,\mathbf{a}}(J)| \leq {B\choose a_1}{B\choose a_2}\prod_{i=3}^{r}{AB \choose a_i} \Pr[uv\in E(R_J)],$
and thus
$$|Y(J)|=\sum_{uv\in E(G)}\sum_{\mathbf{a}\in \mathcal{P}}|Y_{uv,\mathbf{a}}(J)| \leq \sum_{uv\in E(G)}\sum_{\mathbf{a}\in \mathcal{P}}{B\choose a_1}{B\choose a_2}\prod_{i=3}^{r}{AB \choose a_i} \Pr[uv\in E(R_J)] = S\cdot \mathbb{E}[e(R_J)].$$
Now to prove Inequality~\eqref{eq:multipar-YJ}, it suffices to show that
\begin{equation}\label{eq:multipar-eRJ}
    \mathbb{E}[e(R_J)] \leq z +\frac{cv(Z)^{w-\eta}}{A^{w-k}B^{w}}.
\end{equation}

Since deleting one edge from each copy of $K_k$ in $R_J$ yields a $K_k$-free subgraph of $G$, we have $\ex(G, K_2, K_k)\geq e(R_J)-\mathcal{N}(R_J, K_k)$.
Combining this inequality with \eqref{eq:multipar-z}, we obtain
$e(R_J)\leq \ex(G, K_2, K_k)+\mathcal{N}(R_J, K_k)=z+\mathcal{N}(R_J, K_k).$
Taking expectations, we have
\begin{equation}\label{eq:multipar-eRJ+}
    \mathbb{E}[e(R_J)] \leq z+\mathbb{E}[\mathcal{N}(R_J, K_k)].
\end{equation}
It remains to show that $\mathbb{E}[\mathcal{N}(R_J, K_k)]\leq  \frac{cv(Z)^{w-\eta}}{A^{w-k}B^{w}}.$

Let $F_0\in \mathcal{F}$ with $\chi(F_0)=k$.
Then $J$ is $F_0$-free.
Let $W$ be the graph obtained from $K_k$ by adding, for every edge of $K_k$, $r-2$ new vertices and adding all edges so that these $r-2$ vertices together with the two endvertices of the edge form a copy of $K_r$.
All the new vertices are distinct, and no edges are added between vertices corresponding to distinct edges of $K_k$.
Then $v(W)=k+(r-2)\binom{k}{2}=w.$
Moreover, since $\chi(F_0)=k$ and $K_k\subseteq W$, $F_0$ is a subgraph of a sufficiently large blow-up of $W$.
By Lemma~\ref{le:AlSh16-degen}, there exist constants $c_1>0$ and $0<\eta<w$ such that $\mathcal{N}(J,W)\leq c_1v(Z)^{w-\eta}.$

For every copy $K$ of $K_k$ in $G$, let $\mathbf{1}_K$ be the indicator random variable of the event that $K$ is a subgraph of $R_J$.
Then $\mathcal{N}(R_J,K_k)=\sum_{K}\mathbf{1}_K,$ where the sum is taken over all copies $K$ of $K_k$ in $G$.
Hence, we have $$\mathbb{E}[\mathcal{N}(R_J,K_k)]=\sum\nolimits_{K}\Pr[K\subseteq R_J]$$ by linearity of expectation.
Fix a copy $K$ of $K_k$ in $G$.
By the definition of $R_J$, in order that $K$ be contained in $R_J$, we need to choose one vertex from $C_u$ for every $u\in V(K)$, and for every $uv\in E(K)$, we need to choose one vertex from each of $C_{uv}^3, \ldots, C_{uv}^r$. The total number of possible choices of these $k+(r-2)\binom{k}{2}=w$ vertices is $B^k(AB)^{w-k}=A^{w-k}B^w.$
All these choices occur with the same probability $\frac{1}{A^{w-k}B^w}$.
If $K\subseteq R_J$, then the $w$ chosen vertices form a copy of $W$ in $J$.
Let $M_K$ be the number of choices of these $w$ vertices for which $K\subseteq R_J$.
Then $\Pr[K\subseteq R_J]=\frac{M_K}{A^{w-k}B^w},$ so
$$\mathbb{E}[\mathcal{N}(R_J,K_k)]=\sum\nolimits_{K}\Pr[K\subseteq R_J]=\sum\nolimits_{K}\frac{M_K}{A^{w-k}B^w}=\frac{1}{A^{w-k}B^w}\sum\nolimits_{K} M_K.$$

Since $W$ is a fixed graph depending only on $r$ and $k$, there is a constant $c_{r,k}>0$ such that every copy of $W$ in $J$ can arise from at most $c_{r,k}$ choices of the $w$ vertices counted in $\sum_K M_K$.
Thus $\sum_K M_K \leq c_{r,k}\mathcal{N}(J,W).$
Combining this bound with $\mathcal{N}(J,W)\leq c_1v(Z)^{w-\eta}$, we have
$$\mathbb{E}[\mathcal{N}(R_J,K_k)]=\frac{1}{A^{w-k}B^w}\sum\nolimits_{K} M_K\leq \frac{1}{A^{w-k}B^w}c_{r,k} \mathcal{N}(J,W)\leq \frac{c_{r,k}c_1v(Z)^{w-\eta}}{A^{w-k}B^{w}} = \frac{cv(Z)^{w-\eta}}{A^{w-k}B^{w}},$$
where $c=c_{r,k}c_1$ depends only on $T$ and $\mathcal{F}$.
This together with \eqref{eq:multipar-eRJ+} implies \eqref{eq:multipar-eRJ}, and hence \eqref{eq:multipar-YJ}.
The proof of Claim~\ref{cl:multipar-upper} is complete.
\end{proof}

\begin{claim}\label{cl:multipar-lower}
$\ex'(Z, T, \mathcal{F})\geq Sz.$
\end{claim}

\begin{proof}
Let $H$ be a spanning subgraph of $G$ satisfying $\chi(H)\leq k-1$ and $e(H)=z$.
Let $Z_0$ be the spanning subgraph of $Z$ obtained by deleting all edges of $G(B)$.
Then $Z_0+H(B)$ is also a spanning subgraph of $Z$.
We shall prove that $Z_0+H(B)$ is $\mathcal{F}$-free.

We first show that $\chi(Z_0+H(B))\leq k-1$.
Fix a proper $(k-1)$-coloring of $H$, and assign every vertex of $C_u$ the color of $u$.
For each edge $uv\in E(G)$, whether the colors of $u$ and $v$ are distinct or equal, there are at least $k-1-2\geq r-2$ available colors.
So we can assign pairwise distinct colors, different from the colors of $u$ and $v$, to $C_{uv}^3, \ldots, C_{uv}^r$.
Since auxiliary vertices corresponding to distinct edges of $G$ are nonadjacent, this yields a proper $(k-1)$-coloring of $Z_0+H(B)$.
Hence, $\chi(Z_0+H(B))\leq k-1$.

Since $\min_{F\in \mathcal{F}}\chi(F)=k$, the graph $Z_0+H(B)$ is $\mathcal{F}$-free, so $\ex'(Z, T, \mathcal{F})\geq \mathcal{N}(Z_0+H(B),T).$
Moreover, for every edge $uv\in E(H)$, the sets $C_u, C_v, C_{uv}^3, \ldots, C_{uv}^r$ form a complete $r$-partite graph, and thus contain exactly $S$ copies of $T$ with $q$ auxiliary vertices by Claim~\ref{cl:multipar-T-2}.
Therefore,
$\ex'(Z, T, \mathcal{F})\geq \mathcal{N}(Z_0+H(B),T)\geq Se(H)=Sz.$
The proof is complete.
\end{proof}

Let $w=k+(r-2){k\choose 2}$, and $c, C, \eta$ be the constants given by Claim~\ref{cl:multipar-upper}.
Let $A\colonequals n^{v(T)+q}$ and $B\colonequals n^{\beta}$, where $\beta$ is a fixed positive integer depending only on $T, \mathcal{F}, \varepsilon$, chosen so that
\begin{equation}\label{eq:multipar-parameters}
     \beta\eta> 2(w-\eta)+(v(T)+q)(k-\eta)+1, \qquad \beta\varepsilon> 2(v(T)-\varepsilon)+(v(T)+q)(v(T)-\varepsilon-q)+1.
\end{equation}
Define
$$L\colonequals 0, \qquad a\colonequals 0, \qquad b\colonequals S \frac{cv(Z)^{w-\eta}}{A^{w-k}B^{w}} +Cn^{v(T)+q-1}A^{q-1}B^{v(T)}.$$
Then by Claims~\ref{cl:multipar-upper} and \ref{cl:multipar-lower}, we have
$$L+Sz-a \leq \ex'(Z, T, \mathcal{F})\leq \ex(Z, T, \mathcal{F})\leq L+Sz+b.$$
Recall that $v(Z) = O_r(n^2AB)$ and $S=\Theta_T\left(A^{q}B^{v(T)}\right).$
Combining these estimates with \eqref{eq:multipar-parameters}, we have
\begin{align*}
S^{-1}b= &~S^{-1}S \frac{cv(Z)^{w-\eta}}{A^{w-k}B^{w}} +CS^{-1}n^{v(T)+q-1}A^{q-1}B^{v(T)} \\
= &~O\left((n^2AB)^{w-\eta}A^{k-w}B^{-w}\right) +O\left(A^{-q}B^{-v(T)}n^{v(T)+q-1}A^{q-1}B^{v(T)}\right) \\
= &~O\left(n^{2(w-\eta)}A^{k-\eta}B^{-\eta}\right) +O\left(n^{v(T)+q-1}A^{-1}\right) \\
= &~O\left(n^{2(w-\eta)+(v(T)+q)(k-\eta)-\beta\eta}\right) +O\left(n^{v(T)+q-1-(v(T)+q)}\right)= O(n^{-1}) = o(1),
\end{align*}
\begin{align*}
S^{-1}v(Z)^{v(T)-\varepsilon}= &~ O\left(A^{-q}B^{-v(T)}(n^2AB)^{v(T)-\varepsilon}\right) = O\left(n^{2(v(T)-\varepsilon)}A^{v(T)-\varepsilon-q}B^{-\varepsilon}\right)\\
= &~O\left(n^{2(v(T)-\varepsilon)+(v(T)+q)(v(T)-\varepsilon-q)-\beta\varepsilon}\right) = O(n^{-1}) = o(1).
\end{align*}
Thus
$$a+b+v(Z)^{v(T)-\varepsilon}=o(S).$$

Choose a fixed integer $N^{\ast}\geq n_0(k)$ so that $B\geq s_r$ and $a+b+v(Z)^{v(T)-\varepsilon}<S/2$ whenever $n\geq N^{\ast}$, and define $\mathcal{G}_{k,N^{\ast}} \colonequals \{G\in\mathcal G_k\colon\, v(G)\geq N^{\ast}\}.$
Computing $\cut_{k-1}(G)$ remains NP-hard on this class, since an exact polynomial-time algorithm on this class would extend to all of $\mathcal{G}_k$ by exhaustive search on the instances with fewer than $N^{\ast}$ vertices.
Since $v(Z) = O_r(n^2AB)$, $A=n^{v(T)+q}$ and $B=n^{\beta}$, where $\beta$ is fixed once $T, \mathcal{F}, \varepsilon$ are fixed, we have $v(Z)=O(n^{v(T)+q+\beta+2})$.
Hence, $Z$ can be constructed in polynomial time.
Since $\mathcal{P}$ depends only on $T$ and $B\geq s_r$, the formula $S= \sum_{\mathbf{a}\in \mathcal{P}}{B\choose a_1}{B\choose a_2}\prod_{i=3}^{r}{AB \choose a_i}$ shows that $S$ is a polynomial-time computable positive integer.
Now by Corollary~\ref{cor:transfer}, each of the two approximation problems is NP-hard even when the input graph is restricted to $\{Z(G)\colon\, G\in\mathcal G_{k,N^{\ast}}\}$, and hence also for arbitrary host graphs.
This completes the proof of Theorem~\ref{thm:multipar}.

\section{Concluding remarks}
\label{sec:conclu}

In this paper, we have established Conjecture~\ref{conj:AlSh22} for every bipartite graph $T$ and every complete multipartite graph $T$, and for the additional pairs covered by Theorems~\ref{thm:chromatic+2}, \ref{thm:chromatic} and \ref{thm:clique}.
In particular, the conjectured hardness holds for an arbitrary graph $T$ with $e(T)>0$ whenever $\min_{F\in \mathcal{F}}\chi(F)\geq \chi(T)+2$.
Our results apply to both $\ex(G,T,\mathcal{F})$ and $\ex'(G,T,\mathcal{F})$,
and extend the results of Alon, Shapira and Sudakov~\cite{AlSS} for $\ex(G,K_2,\mathcal{F})$ and of Alon and Shikhelman~\cite{AlSh22} for $\ex(G,K_m,K_k)$.

The proofs encode $\cut_{k-1}(G)$ in the maximum number of copies of $T$ in an $\mathcal{F}$-free subgraph of an auxiliary graph $Z=Z(G)$, where $k=\min_{F\in \mathcal{F}}\chi(F)$.
The constructions must satisfy two properties simultaneously:
\begin{itemize}
\item[(1)] An optimal $(k-1)$-colorable spanning subgraph of $G$ must yield an $\mathcal{F}$-free spanning subgraph of $Z$ that contains enough copies of $T$ to establish the desired lower bound on $\ex'(Z, T, \mathcal{F})$.
\item[(2)] Every $\mathcal{F}$-free subgraph $J\subseteq Z$ must satisfy a corresponding upper bound on the number of copies of $T$, thereby giving an upper bound on $\ex(Z, T, \mathcal{F})$.
\end{itemize}
Iterative leaf additions and symmetric rooted extensions, together with suitable blow-ups, provide these properties in the cases studied here.
The random-transversal argument is a key ingredient in establishing the upper bounds.

A natural next question is whether the chromatic hypothesis in Theorem~\ref{thm:chromatic+2} can be weakened to $\min_{F\in \mathcal{F}}\chi(F)\geq \chi(T)+1$ without a condition on edge contraction.
More generally, extending this approach to the full conjecture would require constructions that exploit the absence of a homomorphism $F\to T$ for every $F\in \mathcal{F}$, including cases in which $\min_{F\in \mathcal{F}}\chi(F)\leq \chi(T)$.

\section*{Acknowledgements}

This work was supported by the National Key Research and Development Program of China (No. 2023YFA1010201),
the National Natural Science Foundation of China (No. 12125106, No. 12501492, No. 12501474 and No. 12601663),
the Science and Technology Commission of Shanghai Municipality (No. 22DZ2229014),
the Shaanxi Fundamental Science Research Project for Mathematics and Physics (No. 25JSQ043),
and the Fundamental Research Funds for the Central Universities (No. GK202506024).

\section*{Declaration of competing interests}

The authors declare no competing interests.

\section*{Data availability}

No data were used for the research described in the article.

\section*{AI methodology}

We used ChatGPT to assist with exploring possible graph constructions, checking mathematical arguments, and improving the presentation.
The tool materially affected the iterative leaf addition outlined in Section~\ref{subsec:sketch} and developed in Section~\ref{sec:pf_bipar}.
The authors developed and formalized the constructions and proofs, independently verified the mathematical arguments and references, and prepared the manuscript.
The authors take full responsibility for all statements, conclusions, proofs, and citations presented in this manuscript.

\addcontentsline{toc}{section}{References}
\bibliographystyle{alpha}
\bibliography{Alg_Turan_References}

\end{document}